\documentclass{article}

\usepackage{comment}

\newcommand{\shortlong}[2]{#2}
\shortlong{
\pdfpagewidth=8.5in
\pdfpageheight=11in

\usepackage{ijcai24}

\pdfinfo{
/TemplateVersion (IJCAI.2024.0)
}
\newcommand{\shortversion}[1]{#1}
\newcommand{\longversion}[1]{}
\usepackage[appendix=strip]
{apxproof}

}{
\usepackage{authblk}
\usepackage{fullpage}
\newcommand{\longversion}[1]{#1}
\newcommand{\shortversion}[1]{}
\usepackage[appendix=inline]{apxproof}
}
\usepackage[switch]{lineno}
\shortversion{\linenumbers}

\usepackage{times}
\usepackage{soul}
\usepackage{url}
\usepackage[hidelinks]{hyperref}
\usepackage[utf8]{inputenc}
\usepackage[small]{caption}
\usepackage{graphicx}
\usepackage{amsmath}
\usepackage{amsthm}
\usepackage{booktabs}
\usepackage{algorithm}
\usepackage[noend]{algorithmic}

\newtheorem{theorem}{Theorem}

\usepackage{subcaption}
\usepackage{tikz}
\usetikzlibrary{shapes.geometric}

\usepackage{todonotes}

\newcommand{\equivnd}{\equiv_{\snd}}

\newcommand{\iffl}{if\longversion{ and only i}f }

\newlength{\fullwidth}
\longversion{\setlength{\fullwidth}{\textwidth}}
\shortversion{\setlength{\fullwidth}{.5\textwidth}}

\newenvironment{pfclaim}{\begin{proof}}{\end{proof}}

\usepackage[most]{tcolorbox}
\newcommand{\problemdef}[4]{\setlength\tabcolsep{2pt}
	\begin{tcolorbox}[width = \columnwidth,colback=blue!5!white,colframe=blue!75!black,arc=0pt,outer arc=0pt,boxrule=1.5pt,left =0.5em,right=0em,enhanced,attach boxed title to top center={yshift=-3.8mm,yshifttext=-.9mm},colbacktitle=red!60,
  title=\textsc{#1} #2,
  boxed title style={size=small,colframe=red!50!black}]
  \mbox{}\\[-.5ex]
		\begin{tabular}{ @{\!\!\!}l p{0.825\columnwidth} c }
			\textsf{Input:} & #3 \\[.5pt]
			\textsf{Problem:} & #4
		\end{tabular}
	\vspace{-0.55em}
	\end{tcolorbox}
}
\usepackage{xspace}

\newcommand{\yes}{\textsf{yes}\xspace}

\newcommand{\snd}{{\small \textsf{snd}}}

\newcommand{\vc}{\textsf{vc}\xspace}

\newcommand{\NP}{\textsf{NP}}
\newcommand{\FPT}{\textsf{FPT}\xspace}
\newcommand{\XP}{\textsf{XP}}

\newcommand{\W}[1]{\ensuremath{\textsf{W}[#1]}}
\newcommand{\paraNP}{\textsf{para-NP}}

\usepackage{newfloat}
\usepackage{listings}
\usepackage{amsfonts,amssymb}
\usepackage{mathtools}

\newtheoremrep{theorem}{Theorem}
\newtheorem{definition}{Definition}

\newtheoremrep{lemma}[theorem]{Lemma}
\newtheoremrep{proposition}[theorem]{Proposition}
\newtheoremrep{claim}[theorem]{Claim}
\newtheoremrep{corollary}[theorem]{Corollary}

\title{Defensive Alliances in Signed Networks}
\shortlong{%
\author{%
}
}{
\iftrue
\author[4]{Emmanuel~Arrighi}
\author[1,2]{Zhidan~Feng}
\author[2]{Henning~Fernau}
\author[2]{Kevin~Mann}
\author[1]{Xingqin~Qi}
\author[3]{Petra~Wolf}
\affil[1]{Shandong University, School of Mathematics and Statistic,  264209 Weihai, China}
\affil[2]{Universit\"at Trier, Fachbereich~4 -- Abteilung Informatikwissenschaften, 54286 Trier, Germany}
\affil[3]{LaBRI, CNRS, Université de Bordeaux, Bordeaux INP, France}
\affil[4]{EnsL, Univ Lyon, UCBL, CNRS, Inria,  LIP, F-69342, LYON Cedex 07, France.}
\fi
}

\begin{document}

\maketitle

\begin{abstract}
The analysis of (social) networks and community detection is a central theme in Artificial Intelligence. One line of research deals with finding groups of agents that could work together to achieve a certain goal. To this end, different notions of so-called clusters or communities have been introduced in the literature of graphs and networks. Among these, a \emph{defensive alliance} is a kind of quantitative group structure. However, all studies on alliances so far have ignored one aspect that is central to the formation of alliances on a very intuitive level, assuming that the agents are preconditioned concerning their attitude towards other agents: they prefer to be in some group (alliance) together with the agents they like, so that they are happy to help each other towards their common aim, possibly then working against the agents outside of their group that they dislike. Signed networks were introduced in the psychology literature to model liking and disliking between agents, generalizing graphs in a natural way.
Hence, we propose the novel notion of a defensive alliance in the context of signed networks. We then investigate several natural algorithmic questions related to this notion. These, and also combinatorial findings, connect our notion to that of \emph{correlation clustering}, which is a well-established idea of finding groups of agents within a signed network. Also, we introduce a new structural parameter for signed graphs, the \emph{signed neighborhood diversity}~\snd, and exhibit a \snd-parameterized algorithm that finds one of the smallest defensive alliances in a signed graph.
\end{abstract}

\section{Introduction and Motivation}

Graphs or networks (we use these terms interchangeably) are one of the most important data structures, which can represent complex relations between objects. Mining graph data has become a hot research field\longversion{ in recent years}. \longversion{Some t}\shortversion{T}raditional data mining algorithms have been gradually extended to graph data, such as clustering, classification and frequent pattern mining.  %\todo{Submissions may consist of up to 7 pages of technical content plus additional pages solely for references} 
\longversion{However}\shortversion{Yet}, in the context of graph data, when the nodes of a network may represent persons or groups of persons, up to political entities like countries, one should also consider categories that express, e.g., that two persons like or dislike each other. This has led to the concept of \emph{signed graphs}, introduced in \cite{Hei46,Har5354}. 
Positive edges model trust, friendship, or similarity, while negative edges express distrust, opposition\longversion{ (enmity, hostility)}, or antagonism. 
This model originated from psychology; there, it continued to flourish over decades\longversion{; see \cite{DeSoto1968,Mohazab1985,Christian2007,BruSte2010,Brashears2016,Zorn2022}}\shortversion{ \cite{DeSoto1968,Mohazab1985,Christian2007,Brashears2016,Zorn2022}}. 
The notion of signed graphs has become one of the main modeling tools in modern network science, with  significant applications; e.g., link prediction\longversion{, as in} \cite{Ye2013,Song2015}, social system reasoning\longversion{, as in} \cite{TanCAL2016}, clustering and community detection, see, e.g., \cite{Cad2016,Li2021} and the following discussion.

In classical unsigned graph models of agent interaction, where edges between vertices describe links between these agents, the simplest form of a `perfect group' of agents corresponds to a clique. Community detection then amounts in finding (large) cliques, which is a well-researched area, being equivalent to finding large independent sets or small vertex covers. The corresponding research stretches over decades\longversion{; see}  \cite{BroKer73,EppLofStr2013}. \cite{HaoYMY2014} adapted the notion of a clique (and hence that of a community) to 
signed graphs by requiring a positive edge between any two community members. By way of contrast,  \cite{Li2021} suggested a notion of a signed clique that allows a certain fraction of negative connections within a clique.
A bit more relaxed, \cite{YanCheLiu2007} described a community in a signed network as a dense positive subgraph such that, assuming that the graph has been partitioned into such groups, the  negative between-group relations are also dense. This idea was also used to find antagonistic communities in \cite{ChuLingyang2016,GaoMing2016}.
Put to the extreme\longversion{ case of denseness}, we \shortversion{come to}\longversion{arrive at} the notion of weakly balanced signed graphs discussed below.

We will follow the approach of detecting communities defined by structural properties as exemplified above\shortversion{,} \longversion{in this paper,}lifting notions from unsigned to signed networks, as opposed to approaches based on spectral methods %or centrality measures\longversion{ like page ranking}\todo[color=green]{EA: Isn't centrality structural property? HF: True, how to make this precise?}, 
or inspired by physical statistics, as \shortversion{shown}\longversion{delineated} in \shortversion{\cite{KunLomBau2009,PalDFV2005,YanCheLiu2007}}\longversion{\cite{KunLomBau2009,PalDFV2005,PonLap2006,YanCheLiu2007}}, to give\longversion{ only} a few pointers.%\todo{some refs could be removed here}
\longversion{

Clustering on graph data is also called community detection. Then for unsigned graph, the main purpose is to partition the graph into many clusters such that edges within each cluster are dense while edges between clusters are sparse.} To quantitatively measure the relative density of a cluster, one can assume each node~$v$ in a cluster~$S$ has more neighbors in~$S$ than outside. Such vertex set $S$ is also called a ``defensive alliance'', \longversion{a concept that was first studied}\shortversion{originating} in \shortversion{\cite{KriHedHed2004,SzaCza2001,Sha2004}}\longversion{\cite{FriLHHH2003,Kimetal2005,KriHedHed2004,SzaCza2001,Sha2004}}. Up to now, hundreds of papers have been published on alliances in graphs\longversion{ and related notions}, as \longversion{also certified by}\shortversion{well as} surveys and book chapters\longversion{; see, in chronological ordering,}  \longversion{\cite{FerRod2014a,YerRod2017,OuaSliTar2018,HayHedHen2021}}
\shortversion{\cite{YerRod2017,OuaSliTar2018,HayHedHen2021}}.
\longversion{An overview on the wide variety of applications for alliances in graphs is given}\shortversion{Many applications of alliances in graphs are shown}  in \cite{OuaSliTar2018}, including community-detection problems\longversion{ as described in} \cite{SebLagKhe2012}.%\todo{This sentence reads misplaced in the short version I would rather keep the first part of the sentence giving the overview on the wide variety of applications}
%Possible applications are nicely described in \cite{OuaSliTar2018}, among them also community-detection problems~\cite{SebLagKhe2012}. 

More formally, a vertex set~$S\neq\emptyset$ of an unsigned graph~$G$ is a \emph{defensive alliance} if, for each $v\in S$, $|N(v)\cap S|+1\geq |N(v)\setminus S|$, with $N(v)$ denoting the\longversion{ open} neighborhood of~$v$\longversion{. Hence, another interpretation or intuition behind a defensive alliance is that}\shortversion{, i.e.,} each member of the alliance can be protected by the alliance from any attack outside. Edges indicate close\shortversion{ness}\longversion{ geographic distances}, so that an attack is\longversion{ indeed} possible, as well as a defense. The world is often not that simple, as between different countries, there might be more or less friendly bonds\longversion{ for historical or many other reasons}. Thus, it makes sense to qualify the edges modelling \longversion{geographic }vicinity as `friend or foe'\longversion{ relationships}. \longversion{This is why w}\shortversion{W}e propose to adapt the notion of a defensive alliance to\longversion{wards} signed graphs, which naturally capture these aspects.
%were designed to model these aspects.
Moreover, the cited previous work on antagonistic communities neglected the aspect of defensiveness, 
%Recently, people are increasingly interested in searching group structures from different perspectives in the signed networks. %\cite{Li2021} develop a maximal signed clique in signed networks. Antagonistic Community of signed networks are described in \cite{GaoMing2016,DavidLO2013}, which aims to find two conflicting cohesive communities. Chu~\etal \cite{ChuLingyang2016} generalize this problem to find $k$-Oppositive Cohesive Groups. Note that these oppostive cohesive groups can be seen as clusters in signed network, but it is not necessary a defensive alliance. If $v$ has more ``enemies" outside than ``friends" inside, it can not be protected by the members in the same group. Therefore, it is very necessary to present a new definition of ``defensive alliance" for signed networks. Here, we will investigate an alliance structure for the purpose of defensiveness on a signed network, 
which is more accurate in depicting international relationships\longversion{, see} \cite{Doreian2015}, or social opinion networks\longversion{, refer to} \cite{KunLomBau2009}, etc.

%\todo[inline]{to continue from here}

Classical concepts of clusters, as \cite{CarHar56}, in signed networks %\todo{HF: New para to explain our choice } would 
allow only positive edges between cluster members.
However, this need not reflect real-world situations\shortversion{, where one can}\longversion{. Here, we relatively often} find entities with well-established negative relationships within one alliance. We will give a concrete example from ethnology below. Even though in more formalized political context, alliance members mostly promise to help each other (irrespectively of possible negative preconditions), it might well be that \longversion{some alliance members will be quicker to help than others; }at least for an immediate response to a threat, members with positive preconditions will be far more reliable. Our proposed definition of a defensive alliance takes this into account: it rules out too many frictions within an alliance by requiring that each members has less  negative neighbors in the alliance than positive  ones, but it also requires that each member should face  less  negative neighbors outside the alliance than positive neighbors within to enable defense. %\todo{I added informal def. }
The formal definition will follow in a later section.
We %\todo{adapted from rebuttal} 
also like to mention that group structures similar to the notion of defensive alliances defined in our paper have been considered before in the AI and ML literature.
For instance, \cite{TzeOrdGio2020} describes procedures to discover conflicting groups within a signed network; these groups should satisfy certain properties. The first one is that within a group, the edges are mostly positive, and the edges between two groups are mostly negative. If the groups form defensive alliances, then this property is satisfied. \longversion{Even m}\shortversion{M}ore relaxed than this  is the notion of $k$-oppositive cohesive groups from \cite{ChuLingyang2016} where \longversion{only }a large intra-group cohesion and a larger inter-group opposition is required.
From an algorithmic viewpoint, \longversion{let us stress}\shortversion{note} that the properties of the groups of vertices that we look for are quite explicit, while most approaches to detecting communities in (signed) networks are based on spectral methods\shortversion{:}\longversion{, confer} \longversion{\cite{CucDGT2019,CucSST2021,MerTudHei2016,MerTudHei2019,WanWanSo2022}}\shortversion{\cite{CucDGT2019,CucSST2021,MerTudHei2019}}. Our approach is \longversion{of a }purely combinatorial\longversion{ nature}. \longversion{Yet, i}\shortversion{I}t would be interesting to contrast this with spectral approaches.  We will now \longversion{give a number of}\shortversion{explain three} applications.

%\todo[inline,color=green]{Bubbles}
\paragraph{Echo chambers.}
With the rise of social networks, some emergent behavior has been studied\shortversion{ such as so called echo chambers~\cite{sunstein2001echo,bessi2015science,del2016spreading}}.\longversion{ The one
we are interested in is called echo chambers \cite{sunstein2001echo,bessi2015science,del2016spreading}.}
%An \emph{echo chamber} (often also called a bubble) is an environment which regroups like-minded people with the effect
%of reinforcing pre-existing beliefs and insulate people from rebuttal.
An \emph{echo chamber} (often also called a bubble) is an environment that brings like-minded people together, with the effect of reinforcing existing beliefs and shielding them from opposing opinions. 
Even though the existence of echo chambers at a large scale has been questioned and the role
that the social media plays in those is not fully understood \cite{ross2022echo},
studies found that on both ends of the political spectrum,
small echo chambers exist \cite{ross2022echo,jiang2021social}.
%Those small e
Echo chambers raise two problems. First, combined with group polarization,
this can lead to more extreme beliefs, and can induce a shift from word to action
as has been reported in the case of the ``Incel'' phenomenon~\cite{salojarvi2020incel}.
Secondly, those groups are more prone to spreading fake news and are harder to reach.
For instance, in the context of diseases like COVID-19, this can be a matter of public health \cite{jiang2021social}.
 Echo chambers can be described as a group of people that (mostly) share the same beliefs
(positive edges) but are hardly exposed to opposing ideas (negative edges) which basically corresponds
to our definition of defensive alliances in signed networks.
%With positive and negative interaction, one could find, study and plateform could use counter measure to mitigate the phenomenon~\cite{} As it already exists for Buzzfeed for example.
Hence, automatically detecting (small) echo chambers (alliances) in signed networks, a task that we study in this paper, can be seen as a first step towards solving the problematic consequences observed in this context.\longversion{ Related problems occur in \emph{filter bubbles} \cite{Amr2021,Che2023,Spo2017} and hence our algorithms could be also useful in that context.}
%\todo[inline,color=green]{Bubbles}
%\begin{figure}
    %\centering
    %\includegraphics[width=.9\textwidth]{pictures/Rea54-Figure5.PNG}
   % \caption{The structure of relationships between different sub-tribes in some mountain area in Papua New Guinea; dashed lines indicate hostility and solid lines friendship.}
    %\label{fig:Read-example}
%\end{figure}

\begin{figure}[bt]
    \centering
    \scalebox{.9}{\begin{tikzpicture}
        \tikzset{
  punkt/.style={
    align=left,
  },
}
        %\node [draw, shape=circle] (x1) at ([xshift=3em]b.east) {1};
        \node [draw, shape=circle] (x1) at (0.5,0.4) {1};
        \node [draw, shape=circle] (x2) at (0,3) {2};
        \node [draw, shape=circle] (x3) at (-0.8,1.1) {3};
        \node [draw, shape=circle] (x4) at (-0.8,-0.4) {4};
        \node [draw, shape=circle] (x5) at (-0.5,-2.2) {5};
        \node [draw, shape=circle] (x6) at (-1.1,2) {\textcolor{red}{6}};
        \node [draw, shape=circle] (x7) at (-3.4,0.6) {7};
        \node [draw, shape=circle] (x8) at (-3.6,-2) {8};
        \node [draw, shape=circle] (x9) at (-2.8,3) {\textcolor{blue}{9}};
        \node [draw, shape=circle] (x10) at (-4,3.7) {\textcolor{blue}{10}};
        \node [draw, shape=circle] (x11) at (-4,2) {\textcolor{red}{11}};
        \node [draw, shape=circle] (x12) at (-5.2,-1.5) {12};
        \node [draw, shape=circle] (x13) at (-6.3,3.3) {\textcolor{blue}{13}};
        \node [draw, shape=circle] (x14) at (-6.2,-2.7) {\textcolor{blue}{14}};
        \node [draw, shape=circle] (x15) at (-7.5,2) {\textcolor{red}{15}};
        \node [draw, shape=circle] (x16) at (-7.5,-0.5) {\textcolor{red}{16}};

        \path (x1) edge[] (x2);
        \path (x1) edge[dashed] (x3);
        \path (x2) edge[dashed] (x3);
        \path (x1) edge[dashed] (x4);
        \path (x3) edge[] (x4);
        \path (x1) edge[dashed] (x5);
        \path (x2) edge[dashed] (x5);
        \path (x1) edge[dashed] (x6);
        \path (x2) edge[dashed] (x6);
        \path (x3) edge[] (x6);
        \path (x3) edge[] (x7);
        \path (x5) edge[] (x7);
        \path (x7) edge[] (x6);
        \path (x3) edge[] (x8);
        \path (x4) edge[] (x8);
        \path (x8) edge[] (x6);
        \path (x7) edge[] (x8);
        \path (x2) edge[dashed] (x9);
        \path (x5) edge[] (x9);
        \path (x9) edge[dashed] (x6);
        \path (x2) edge[dashed] (x10);
        \path (x9) edge[] (x10);
        \path (x11) edge[] (x6);
        \path (x7) edge[] (x11);
        \path (x8) edge[] (x11);
        \path (x9) edge[dashed] (x11);
        \path (x10) edge[dashed] (x11);
        \path (x1) edge[dashed] (x12);
        \path (x12) edge[] (x6);
        \path (x7) edge[] (x12);
        \path (x8) edge[] (x12);
        \path (x11) edge[] (x12);
        \path (x13) edge[dashed] (x6);
        \path (x13) edge[] (x7);
        \path (x13) edge[] (x9);
        \path (x13) edge[] (x10);
        \path (x13) edge[dashed] (x11);
        \path (x5) edge[] (x14);
        \path (x8) edge[dashed] (x14);
        \path (x12) edge[dashed] (x14);
        \path (x13) edge[] (x14);
        \path (x1) edge[] (x15);
        \path (x2) edge[] (x15);
        \path (x5) edge[dashed] (x15);
        \path (x9) edge[dashed] (x15);
        \path (x10) edge[dashed] (x15);
        \path (x11) edge[dashed] (x15);
        \path (x12) edge[dashed] (x15);
        \path (x13) edge[dashed] (x15);
        \path (x1) edge[] (x16);
        \path (x2) edge[] (x16);
        \path (x5) edge[dashed] (x16);
        \path (x6) edge[dashed] (x16);
        \path (x11) edge[dashed] (x16);
        \path (x12) edge[dashed] (x16);
        \path (x13) edge[dashed] (x16);
        \path (x14) edge[dashed] (x16);
        \path (x15) edge[] (x16);
        
\longversion{\begin{scope}[xshift=3cm]
            \node[punkt]
      (legend) { 1. Gaveve \\
                    2. Kotuni \\
                    3. Ove  \\
                    4. Alikadzuha \\
                    5. Nagamiza \\
                    \textcolor{red}{6. Gahuku} \\
                    7. Masilakidzuha \\
                    8. Ukudzuha \\
                    \textcolor{blue}{9. Notohana}\\
                    \textcolor{blue}{10. Kohika} \\
                    \textcolor{red}{11. Gehamo} \\
                    12. Asarodzuha \\
                    \textcolor{blue}{13. Uheto} \\
                    \textcolor{blue}{14. Seu've} \\                    
                    \textcolor{red}{15. Nagamidzuha} \\
                    \textcolor{red}{16. Gama} };		
	\end{scope}}
    
    \end{tikzpicture}}
    \caption{The structure of relationships between different sub-tribes in some mountain area in Papua New Guinea; dashed lines indicate hostility and solid lines friendship. Two defensive alliances as observed in the country are colored blue and red. Other alliances, like $\{7,8\}$, have not been observed\longversion{, possibly, as also $\{7\}$ or $\{8\}$ is already an alliance; in contrast, $\{3,4\}$ might be more interested in building an alliance together}.}
    \label{fig:Read-example}
\end{figure}
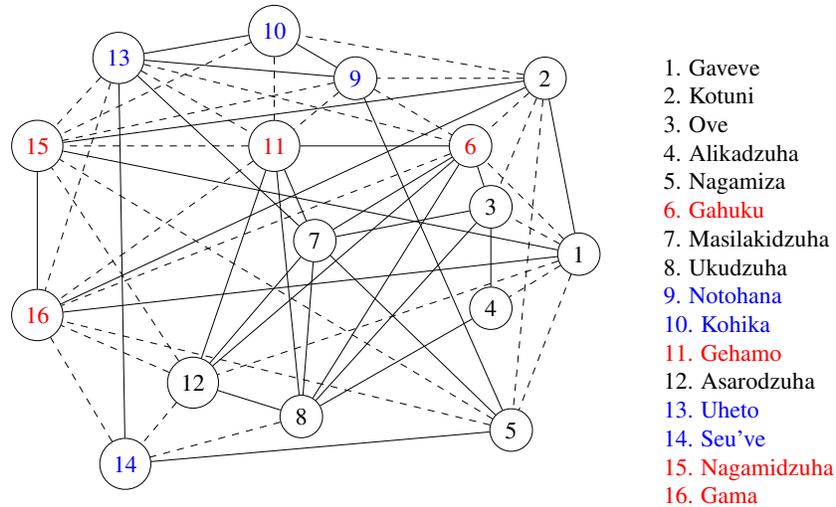

\paragraph{Ethnology.}
\longversion{We are now explaining an example from the social science (ethnographic) literature.}
\cite{Rea54} has described the (sub)-tribal structures of some regions in the central highlands of New Guinea. In particular, \autoref{fig:Read-example} (adapted from that paper) 
is derived from two kinds of \shortversion{geographical}\longversion{physical (geographical)}  
connections between these 16 tribes: \emph{hina} (signalling friendship) % or identification) 
and \emph{rova} (meaning enmity or opposition). 
%\cite{HagHar83} use it to illustrate a clusterable signed graph, where clustering a signed network means to form clusters such that positive edges are existing inside and negative edges are existing between clusters.
%They looked at this data and found that the graph can be partitioned into three sets that `nearly' verify being 3-balanced. However, one of the clusters (consisting of four tribes), among which the Gama, can not protect themselves from attacking outside, because their size is very small, and they have so many enemies outside. That means, only being in a cluster does not mean ``safe". Fortunately, if we take the union of two of these groups, then this part can protect themselves, i.e., form a defensive alliance as we defined in this work. Thus, to coincide realistic purpose, a new graph structure which is different from ``cluster" need to be proposed for signed network. 
Read mentioned that four tribes, among which the Uheto (vertex 13), often form an alliance against four other tribes, among which there are the Gama (16). Both alliances, the one with the Uheto and the one with the Gama, form a defensive alliance as discussed in this paper. They are shown in different colors in \autoref{fig:Read-example}.\longversion{ More precisely, the Uheto ally with the Notohana, the Kohika and the Seu've, while the Gama form an alliance with the Nagamidzuha, the Gehamo, and the Gahuku. The Uheto have friendly relationships with all three allies, while (for geographical reasons) the  Seu've only have friendly relationships with the Uheto. The Gama coalition looks weaker in a sense, but still forms a defensive alliance according to our definition. For instance, the Gama have hina bonds with the Nagamidzuha but rova bonds with both the  Gehamo and the Gahuku. A similar situation is found with the Nagamidzuha, only having one enmity relation less. Similarly, the Gehamo and the Gahuku pair up in friendship, while the Gehamo have  rova bonds with the other two and the Gahuku only with the Gama for some geographical reasons.}  
\cite{HagHar83} also looked at this data and found that the graph can be partitioned into three sets that `nearly' verify being `3-balanced'. Yet, if we take the union of two of these groups, then this part is a defensive alliance, as it is the third group\shortversion{;}\longversion{. All} these alliances are different from the ones observed by Read and discussed above. Rather, it seems to be that in `real life', only smaller alliances are formed\longversion{, possibly due to  difficulties that stone-age cultures faced in highland terrain}.
These aspects cannot be covered by the notions of clusters\shortversion{, see}\longversion{ as proposed and investigated in} \cite{HagHar83}, as we find negative edges between tribes in one group. Our definition tolerates some frictions within an alliance; this is the key to a formal understanding of such groupings found in real-world\longversion{ social network} situations. Read reported that these alliances were stable over time. \shortversion{Can we}\longversion{It would be interesting to} detect \longversion{such social structures not only with hindsight, after studying these groups and their behavior over a long period of time, but to be able to even}\shortversion{or} predict\longversion{ these} alliances based on knowledge about \longversion{the }mutual relationships\shortversion{?}\longversion{ of these tribes. Can we do this algorithmically?} 

\paragraph{Politics.}
Let us underline again that edges in a signed network to be analyzed concerning its alliance structure  often reflects some geographic or mental aspects of vicinity. %\todo{Added for rev. 6} 
In particular, this models the assumption that allies can only help each other (immediately) when they are neighbored.
%But there is also another interpretation of a missing edge: namely, that of a completely neutral relationship. Such an interpretation makes in particular sense in larger social networks, where it is clear that geographic aspects would only exist between a very limited number of participants.\todo{More added to come closer to AI; do we have an example from literature to back this?}
%\todo[inline,color=magenta]{These neutral edges sounds more like there should be a third color for edges, as we don't consider this setting I find the discussion more convincing without this paragraph}
%\todo[inline]{I left out the monk example, as here our analysis is not that convincing and is not reflecting what actually happened.}
In modern times, with the advent of rockets, edges can reflect the willingness to help or to attack. \cite{Doreian2015} empirically analyzed the \emph{Correlates of War} data in the period from 1946 to 1999, based on a signed networks of countries. They showed that  the occurrence of negative relationships in international alliances is even unavoidable. The expected trend of network evolution towards a complete balanced state as expressed in \cite{Har61} is not convincing\longversion{, even worse,}\shortversion{ as} the imbalance \longversion{dramatically }increases over time. Hence, defensive alliance theory appears to be more reasonable as it gives a good  explanation on  the persistent \longversion{(and even increasing) }existence of negative edges between allies. It could provide a new perspective on how to understand the formation of alliances as appearing in the real world. %\todo{Bismarck under comment}
\longversion{The idea of defensive alliances (as formalized above) can be also found in another example from military history. \cite{HeaSte73} looked at several propositions of the multipolar world faced in the European context of the \emph{five powers} between 1870 and 1881. This expression summarized the leading military powers of that time, which were (in alphabetical order) Austria (the Habsburg Empire), France, Germany, Russia and the United Kingdom. 
The German Chancellor Bismarck was said to have expressed the political strategic idea that a power would be safe if it was allied with  at least two among the other four other main powers.\footnote{More historical background information is provided in \url{https://de.wikipedia.org/wiki/Bündnispolitik\_Otto\_von\_Bismarcks}.} 
Notice that this is very much the basic condition of a defensive alliance. Given the fact that Germany had one `fixed enemy' over that time period (namely, France after the French-Prussian war), both France and Germany tried to find partners among the other powers. First, Germany managed to build up the 3-emperors-coalition with Austria and Russia, and there were quite some tensions between France and the UK because of their colonies in Africa. After Bismarck left office, the picture gradually changed towards the alliances seen at the beginning of the Great War. This example is also discussed in Chapter~5 of \cite{EasKle2010}, taking also Italy into account, apart from Austria, France, Germany, Great Britain and Russia. Apart from being a (small) example from history, this also shows a potential application of our algorithmic solutions: namely, that of finding strategic decisions for building alliances. This does not apply only to a military context, but could be also useful, for instance, in the context of companies. We will describe such a scenario in the following paragraph.

\paragraph{Economics.} We can also apply our theory in market economy. \longversion{Apart from implementing ideas of trust and distrust directly, as suggested in \cite{GuhKRT2004}, we can also think of the following scenario, going beyond purely cohesive groups as in \cite{Hil2017}.
}Here, the vertices of a graph represent different companies and positive edges could express tight relations between the companies, often reflected in mutual investments, up to the point where shares are exchanged, while negative edges exist between competitors on a certain market segment, or when companies show a hostile behavior by forcing others into unfavorable contracts. 
No edges between two companies means that they do not work on the same market segment.
Your life as a company is very hard if you are surrounded by many competitors. A survival strategy could be (if money permits) to buy some of the competitors (or, to come to some agreements with them) to alleviate the burden of competition. This explains the interest of companies in some form of alliance building. Such alliances have then not only a defensive character, but can be viewed as monopolies.
On the other side, consumers and possible (groups of) nations are interested in a working market segment with several competitors. Hence, they are interested in detecting monopolies. In this context, allowing negative edges within a monopoly can also be seen as a way of tolerating misinterpretations on the side of the anti-trust regulation authorities, as some links between companies might have been misqualified as hostile.
Although our model is quite special, the hardness results can be interpreted in a way that it is not that easy to find monopolies, in particular in the real world of economy, with many holdings and other interwoven company structures that are hard to analyze.}

%\todo[inline,color=magenta]{Do you have a concrete applications of defensive alliance on signed graphs in the area of AI?}

\section{Main Contributions}
Our main conceptual contribution is the introduction of a notion of \emph{defensive alliance} for signed graphs, which has been only defined for unsigned graphs so far. Based on our definition, we investigate several algorithmic questions.
\begin{itemize}
    \item \textit{Alliability}: Given a signed graph, does there exist a defensive alliance in this graph? Possibly one that contains a prescribed vertex~$v$? While this question has a trivial \yes-answer on unsigned graphs, \longversion{we can prove that even this basic question}\shortversion{it} becomes \NP-complete for signed graphs.  \longversion{We complement our finding by exhibiting}\shortversion{We also give} polynomial-time algorithms\longversion{ for \textsc{Defensive Alliability}} on several graph classes.
    \item \textit{Defensive Alliance Building}: \longversion{In view of the described hardness, }\longversion{one might wonder if one can}\shortversion{Can we} turn a given vertex set into a defensive alliance by converting as few enemy relations into friendships as possible\shortversion{?}\longversion{.} Interestingly, this \longversion{\textsc{Defensive Alliance Building} }problem can be solved in polynomial time.
%\todo[inline]{From the interpretation of missing edges meaning neutrality it would be more meaningful to turn non-edges to positive or negative ones to neutral.}
\item \textit{Finding Small %est 
Alliances}:
%\todo{finding small alliances? \emph{the smallest} sounds like we return a set of guaranteed smallest alliences} 
As we show this question to be  \NP-complete%\todo{As we show that this question... Sounds like its a standing fact from the literature}% (and hard to approximate)\todo[color=magenta]{Do you show this? HF: implicitly so far}
, we investigate aspects of parameterized complexity. We prove that the task of finding smallest alliances remains hard\longversion{ (technically speaking, \W{1}-hard)} when parameterizing by the solution-size parameter (the intended size of the alliance) or by the treewidth of the underlying graph, but it becomes tractable when parameterized by \emph{signed neighborhood diversity}, which is a new structural parameter that we introduce in this paper and that we believe to be useful also for other problems on signed graphs.
\item \textit{Small Alliances in Special Signed Graphs}: The hardness results motivate us to study smallest defensive alliances in several classes of signed graphs, where we can get a good combinatorial understanding. Particularly interesting are \emph{balanced graphs}, as they relate to important network analysis questions like \emph{correlation clustering}. 
\end{itemize}

From a practical perspective, %\todo{HF: paragraph added} 
these questions can be well motivated, for instance as follows:
Having understood the positive or negative relationships in a network of tribes, it would be interesting to know if one can expect to find alliances among these tribes, leading to the question of \textsc{Defensive Alliability}. Related to this is the question to find a small group of tribes (possibly including a chosen tribe) that forms an alliance. If no alliance can be found, then it might be an idea to ask for diplomatic activities that might change the positive or negative preconditions so that a certain group can form an alliance (\textsc{Defensive Alliance Building}). To give an example on a larger scale, without building up the German-French friendship in Europe after WW2, the formation of the European Union would be hardly imaginable.

%As often, we will use the terms \emph{graph} and \emph{network} interchangeably, with the same meaning.
In the following, we will first present the necessary concepts formally. Then, we exhibit some combinatorial studies of these concepts, before introducing some algorithmic problems formally and finally discussing their complexity. 

\section{Formal Concept Definitions}

\begin{toappendix}
We assume some basic knowledge in classical graph theory on the side of the reader. An \emph{(unsigned) graph}~$G$ can be specified as a tuple $(V,E)$, where  $V$ is the vertex set and $E$ the edge set of $G$. More formally, $E\subseteq\binom{V}{2}$, i.e., each edge is a two-element subset of~$V$. $G'=(V',E')$ is a subgraph of $G=(V,E)$ \iffl $V'\subseteq V$ and $E'\subseteq E$; it is an induced subgraph, also denoted as $G[V']$, if $E'=E\cap \binom{V'}{2}$. 

Here, we give a list of basic concepts from classical graph theory which are used in our paper. A graph is said to be embeddable in the plane, or \emph{planar}, if it can be drawn in the plane such that its edges, represented as continuous curves in the plane, intersect only at their endpoints. A graph $G$ is \emph{chordal} if it has a chord for every cycle of length no less than four. A \emph{clique} of a graph $G$ is a complete induced subgraph of $G$. A \emph{vertex cover}~$C$ of a graph $G$ is a subset of vertices of $G$ such that every edge of $G$ has at least one endpoint in $C$. The \emph{vertex cover number} of $G$ is the size of the smallest vertex cover for $G$. An \emph{Independent Set} of a graph $G$ is a subset of vertices such that no two vertices in the subset are adjacent. Clearly, the complementary set of a vertex cover set is an independent set.

%\todo[inline]{HF: thanks for the additions. Notice our conventions (implicit) that (only) problem names are written with small caps, but properties etc. not. It is good tradition to emphasize them (\emph{using emph} in \LaTeX) to highlight its first use or definition. Also, I would not overdo with the definition environment and rather write the well-known notions in a couple of paragraphs.}
%\begin{definition}%[\textsc{Planar}]
    %A graph $G$ is \emph{planar} if it can be drawn in a plane without graph edges crossing.\todo{a rather intuitive explanation, leaving open a lot.}
%\end{definition}

\begin{definition}%[\textsc{Tree Decomposition}]
    A \emph{tree decomposition} of a graph $G$ is a pair $\mathcal{T}=(T, \{X_t\}_{t\in V(T)})$, where $T$ is a tree whose node $t$ is assigned a vertex subset $X_t\subseteq V(G)$, called a bag, satisfying the following:
    \begin{itemize}
        \item [1)] $\bigcup_{t\in V(T)}X_t=V(G)$;
        \item [2)] for each edge $uv\in E(G)$, there is some node $t$ of $T$ such that $u\in X_t, v\in X_t$;
        \item [3)] for each $u\in V(G)$, the set $T_u=\{t\in V(T)\mid u\in X_t\}$ induces a connected subtree of $T$.
    \end{itemize}
\end{definition}

The \emph{width} of tree decomposition $\mathcal{T}$ is given by $\mathop{max}_{t\in V(T)}|X_t|-1$. The \emph{treewidth} of a graph $G$ is the minimum width over all tree decompositions of $G$. For the sake of designing algorithms, it is convenient to think of nice tree decompositions which can be easily transformed by a given tree decomposition with the same width and a linear number of vertices (\cite{CygFKLMPPS2015}).
\begin{definition}%[\emph{Nice Tree Decomposition}]
    A \emph{nice tree decomposition} $\mathcal{T}=(T, \{X_t\}_{t\in V(T)})$ is a tree decomposition if each node $t\in V(T)$ falls into one of the following categories:
    \begin{itemize}
        \item [1)] Leaf node: $t$ is a leaf node of $T$ and $X_t=\emptyset$;
        \item [2)] Introduce node: $t$ has exactly one child $t'$ such that $X_t=X_{t'}\cup \{v\}$  for some $v\notin X_{t'}$;
        \item [3)] Forget node: $t$ has exactly one child $t'$ such that $X_t=X_{t'}\setminus \{v\}$  for some $v\in X_{t'}$;
        \item [4)] Join node: $t$ has exactly two children $t'$ and $t''$ such that $X_t=X_{t'}=X_{t''}$.
    \end{itemize}
\end{definition}

%As we classify the edges into positive and negative edges on signed graphs, we can obtain many different signed graphs from one underlying unsigned graph, in terms of the different configurations of the edge set. Here, our classification for graph classes is consistent with the underlying graphs, with default to arbitrary positive and negative edge configuration, unless specifically stated. For instance, we will speak of a signed $P_n$ to refer to a signed graph where the underlying graph is a path on $n$ vertices; similar to the other graph classes.

%\todo[inline]{Add more defs that we use throughout the paper, also clique or independent set.}
\end{toappendix}

A \emph{signed network} is a triple $G=(V,E^+,E^-)$, where $V$ is a finite set of vertices and $E^+\subseteq \binom{V}{2}$ is the \emph{positive edge} set and $E^-\subseteq \binom{V}{2}$, with $E^+\cap E^-=\emptyset$, is the \emph{negative edge} set. We call $(V,E^+\cup E^-)$ the \emph{underlying (unsigned) graph} of~$G$.
For $v\in V$, $N^+(v)=\{u\in V\mid uv\in E^+\}$ are the \emph{positive neighbors} and $N^-(v)=\{u\in V\mid uv\in E^-\}$ are the \emph{negative neighbors} of~$v$. Accordingly, $\deg^+(v)=|N^+(v)|$ and $\deg^-(v)=|N^-(v)|$ denote the positive and negative degree of vertex~$v$, respectively. We use subscripts to restrict the considered neighborhood to a subset of vertices. For instance, if $X\subseteq V$, then $N_X^+(v)=N^+(v)\cap X$ and $\deg^+_X(v)=|N^+_X(v)|$. \shortversion{Also,}\longversion{Similarly to the unsigned case,} we use $\delta^-(G)(\delta^+(G))$ to denote the smallest negative (respectively, positive) degree in~$G$.
If $X \subseteq V$, then we denote by $\overline{X} = V \setminus X$ the complement of~$X$ with respect to~$V$. \longversion{As in the unsigned case, w}\shortversion{W}e also define the \emph{induced signed graph} $G[X]\coloneqq(X,\{e\in E^+\mid e\subseteq X\},\{e\in E^-\mid e\subseteq X\})$.
We distinguish between friend and enemy relationships and define a \emph{defensive alliance}, or DA for short, of a signed \shortversion{graph}\longversion{network as follows}.%\todo{The following is a novel concept, therefore I used the definition environment.}

\begin{definition}
    A non-empty set~$S$ of vertices of a signed network $G=(V,E^+,E^-)$ is called a \emph{defensive alliance} \iffl: %for each $v\in S$, 
    \begin{enumerate}
        \item $\forall v\in S\colon\deg_S^+(v)+1\geq \deg_S^-(v)$ and
        \item $\forall v\in S\colon\deg_S^+(v)+1\geq \deg_{\overline S}^-(v)$.
    \end{enumerate}
\end{definition}

The first condition expresses that, within an alliance, there should not be more `natural enemies' than friends for each member of the alliance. This models a certain stability aspect of the alliance. It also prevents an alliance from being over-stretched by internal conflicts and rather makes sure that solidarity within the alliance is strong enough so that `natural enemies' within an alliance are at least staying neutral in case of an external attack. The second condition is taken over from the traditional model of defensive alliance in an unsigned graph. It says that each member of an alliance should have at least as many defenders from within the alliance (including itself) than it has enemies outside the alliance. Notice that `natural friends' outside the alliance are considered harmless and will not participate in attacks\shortversion{.}\longversion{, or more explicitly, they are not counted in and are expected to stay neutral, as (if the first condition is met) also `natural enemies' within the alliance are expected to stay neutral.} Both conditions together can also be interpreted as a signed analogue to the idea to maximize the minimum degree within a community, as proposed (e.g.) in \cite{SozGio2010}.
We illustrate these concepts in \autoref{fig:Read-example}\longversion{ and \autoref{fig:example}}. Following the notation introduced for unsigned graphs in \cite{KriHedHed2004,Sha2004}, 
we will denote the size of the smallest defensive alliance in~$G$ by $a_{sd}(G)$. The index $sd$ reminds us that we work on \underline{s}igned graphs and consider \underline{d}efensive alliances.

\longversion{\begin{figure}[bt]
    \centering
    
\begin{subfigure}[t]{.44\fullwidth}
    \centering
    \begin{tikzpicture}
        \tikzset{every node/.style={fill = white,circle,minimum size=0.05cm}}
        
         \node[draw,label=left:$v_2$] (x1) at (0,0) {};
        \node[draw,label=left:$v_3$] (x2) at (0,-2) {};
        \node[draw,label=left:$v_1$] (x3) at (-0.9,-1) {};
        \node[draw,label=left:$v_4$] (x4) at (1,-1) {};
        \node[draw,label=-130:$v_5$] (x5) at (1.6,-1.8) {};
        \node[draw,label=160:$v_6$] (x6) at (1.5,-0.2) {};

        \path (x3) edge[] (x1);
        \path (x3) edge[] (x2);
        \path (x2) edge[] (x1);
        \path (x1) edge[] (x4);
        \path (x2) edge[] (x4);
        \path (x1) edge[] (x6);
        \path (x2) edge[] (x5);
        \path (x4) edge[] (x5);
        \path (x4) edge[] (x6);
        \path (x5) edge[] (x6);

    \end{tikzpicture}
    \subcaption{Unsigned graph: $\{v_5,v_6\}$ is a minimum defensive alliance of size $2$; $\{v_1,v_2,v_3\}$ is a defensive alliance of size $3$.}
    \label{fig:ex_unsigned}
\end{subfigure}
\quad
\begin{subfigure}[t]{.44\fullwidth}
    \centering
    \begin{tikzpicture}
        \tikzset{every node/.style={fill = white,circle,minimum size=0.05cm}}

         \node[draw,label=left:$v_2$] (x1) at (0,0) {};
        \node[draw,label=left:$v_3$] (x2) at (0,-2) {};
        \node[draw,label=left:$v_1$] (x3) at (-0.9,-1) {};
        \node[draw,label=left:$v_4$] (x4) at (1,-1) {};
        \node[draw,label=-130:$v_5$] (x5) at (1.6,-1.8) {};
        \node[draw,label=160:$v_6$] (x6) at (1.5,-0.2) {};

        \path (x3) edge[dashed] (x1);
        \path (x3) edge[dashed] (x2);
        \path (x2) edge[] (x1);
        \path (x1) edge[dashed] (x4);
        \path (x2) edge[] (x4);
        \path (x1) edge[dashed] (x6);
        \path (x2) edge[dashed] (x5);
        \path (x4) edge[dashed] (x5);
        \path (x4) edge[] (x6);
        \path (x5) edge[] (x6);

    \end{tikzpicture}
    \subcaption{Signed graph: $\{v_6\}$ is a minimum defensive alliance; $\{v_1,v_3\}$ is a defensive alliance of size $2$; $\{v_1,v_2,v_3\}$ is no longer a defensive alliance. }
    \label{fig:ex_signed}
\end{subfigure}
    \caption{Defensive alliances on unsigned versus signed graphs, where dashed or solid lines mean negative or positive edges, respectively.}
    \label{fig:example}
\end{figure}
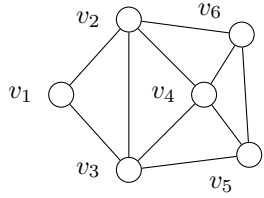
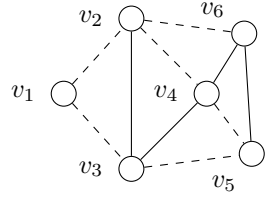}

We will also look into special classes of signed graphs. Often, these are defined via properties of the underlying unsigned graph, for instance, the classes of signed paths, cycles, subcubic, or complete graphs.
But there are also very interesting classes of signed graphs that have no counterpart from the unsigned perspective. 
%Cartwright and Harary 
\cite{CarHar56} defined that a signed graph is \emph{balanced} if each of its cycles contains an even number of negative edges. \cite{Dav67} extended this notion; he called a signed graph a \emph{clustering} or \emph{weakly balanced} if there are no cycles in the graph with exactly one negative edge. He presented the following characterization: a signed graph is weakly balanced  \iffl its vertices can be split into $k\geq 1$ groups (clusters) \longversion{such that}\shortversion{where} edges connecting vertices within groups are positive and edges connecting vertices of different groups are negative. Such a signed graph is also called \emph{$k$-balanced}. \longversion{In other words, %\todo{2 sentences of explanations added for rev. 2} 
in a $k$-balanced signed graph $G=(V,E^+,E^-)$, the positive graph $G^+=(V,E^+)$ contains at least $k'\geq k$ connected components $C_1,\dots,C_{k'}$ (that are not necessarily cliques), while the negative graph $G^-=(V,E^-)$ is $k$-colorable, such that each of the $k$ color classes can be formed by the union of some of the vertex sets $C_1,\dots,C_{k'}$. These color classes correspond to the groups of the $k$-balanced partition of~$G$. }Davis also showed that a complete signed graph is weakly balanced \longversion{if and only if}\shortversion{iff} none of its triangles
has exactly one negative edge. How to partition a complete signed graph best possible to a clustering partitioning has been investigated extensively since the publication of the ground-breaking paper on \textsc{Correlation Clustering} by \cite{BanBluCha2004}\longversion{, being equivalent to \textsc{Cluster Editing} on unsigned graphs, but clearly extendible towards non-complete signed graphs}. 

\section{Combinatorial Prelude}
%\todo[color=magenta,inline]{The order of presented results seems to not follow the order given in the introduction.}
We begin by discussing some observations on the size of a smallest defensive alliance $a_{sd}$  before we continue with more algorithmic results, as this also shows some structural insights.
By the definitions and a chain of inequalities, we \shortversion{see}\longversion{obtain our first result}:

\begin{lemmarep}\label{cor_v}
If \longversion{a vertex }$v$ is in a \longversion{defensive alliance}\shortversion{DA}, then $\deg^{+}({v})+1 \geq \left\lceil \frac{\deg^{-}({v})}{2} \right\rceil.$
\end{lemmarep}

\begin{proof}
For any vertex $v \in S$, 
\begin{align*}
\deg^{+}({v})+1 &{}\geq\deg_S^+(v)+1\\&{}\geq \max\{\deg_S^-(v), \deg_{\overline S}^-(v)\}\\&{}\geq \left\lceil\frac{\deg^{-}(v)}{2} \right\rceil\,.
\end{align*}
\end{proof}
\longversion{\noindent}This simple lemma already has a nice consequence for finding alliances of size at most~$k$.

\begin{corollary}
If there is a vertex $v$ with $\deg^-(v)\geq 2k+1$ in a signed graph~$G$, then $v$ cannot be in any \longversion{defensive alliance}\shortversion{DA} of size\shortversion{~$\leq k$.}\longversion{ at most~$k$ in~$G$.} 
\end{corollary}

Next, we give characterizations of small defensive alliance numbers for signed graphs. \longversion{For the unsigned case, Propositions~1 and~3 in~\cite{KriHedHed2004} describe similar characterizations.}\shortversion{\cite{KriHedHed2004} describe similar characterizations for the unsigned case. As we focus on algorithmic exploits, we present these in \autoref{alg1}.}
\begin{toappendix}
\begin{theorem}\label{thm:asd=1}
 Let $G$ be a signed graph. Then,
 \begin{itemize}
  \item [1)] $a_{sd}(G)=1$ \iffl $\exists v\in V(G):\deg^-(v)\leq 1$.
  \item [2)] $a_{sd}(G)=2$ \iffl $\delta^{-}(G)  \geq 2$ and there exist two adjacent vertices $v, u\in V(G):\deg^-(v)=\deg^-(u)= 2$.
\end{itemize}
\end{theorem}

\begin{proof}
ad 1)~Obviously, a vertex $v$ can defend itself \iffl $\deg^-(v)\leq 1$. 
ad 2)~Except for the case of $a_{sd}(G)=1$, we know the minimum negative degree of~$G$ is at least~2. Assume that we have a defensive alliance~$S$ consisting of two vertices $v, u\in V(G)$. Case one is that $v, u$ are connected positively, then for each vertex, there are at most two negative connections outside of $S$; Case two is that $v, u$ are connected negatively, then for each vertex, there is at most one negative connection outside of $S$. Combined with $\delta^{-}(G) \geq2$, we conclude that the negative degree of $v, u$ is exactly 2, regardless of whether they are positively or negatively connected. For the converse direction, if $\delta^{-}(G) \geq 2$, then from 1) we know $a_{sd}(G) \geq 2$. We can verify that any two adjacent vertices $v, u\in V(G)$ with $\deg^-(v)=\deg^-(u)= 2$ can form a defensive alliance, so that $a_{sd}(G)=2$. 
\end{proof}

 From these combinatorial results, we can conclude the following for signed trees, paths and cycles, keeping in mind that any leaf in a tree forms a minimum defensive alliance.
\begin{corollary}\label{cor:trees}
For any signed tree graph $T$\longversion{, in particular for paths}, $a_{sd}(T)=1$.\end{corollary}

\begin{corollary}\label{cor:cycles}
For any signed cycle graph $C$, if there exists a positive edge, then $a_{sd}(C)=1$; otherwise, $a_{sd}(C)=2$. 
\end{corollary}

\autoref{cor:trees} clearly also extends to forests, but more interestingly, we can also extend \autoref{cor:cycles}.
Recall that a graph is \emph{unicyclic} if it contains exactly one cycle as a subgraph.

\begin{corollary}\shortversion{A unicyclic signed graph~$G$ obeys}\longversion{Let $G$ be a unicyclic signed graph. Then,} $a_{sd}(G)\leq 2$. Moreover, $a_{sd}(G)=2$ \iffl $G$ is a cycle without a positive edge.
\end{corollary}

\begin{proof}
If $G$ is unicyclic but $\delta(G)=1$, then any vertex of degree one forms a defensive alliance. If $G$ is unicyclic but $\delta(G)>1$, then $G$ must be a cycle, so that the previous corollary applies.
\end{proof} 
\end{toappendix}
Signed graph with maximum degree at most three are called signed subcubic graphs. \shortversion{They include paths and cycles.}
\begin{toappendix}
\begin{theorem}\label{thm:subcubic-DA}
There is a polynomial-time algorithm that decides alliability of any signed subcubic graph~$G$ and if so, it determines $a_{sd}(G)$.%\todo[color=green]{Move it to the Section "alliability"?}
\end{theorem}    
\end{toappendix}
\shortversion{

}\longversion{\noindent We formulate this theorem in an algorithmic fashion, but its proof mainly gives combinatorial insights.}%\shortversion{ (\autoref{alg1})}.
\shortversion{
\begin{algorithm}[tb]
\caption{Minimum DA for signed subcubic graphs}\label{alg1}
\begin{algorithmic}\small 
\REQUIRE A signed subcubic graph $G=(V, E^+,E^-)$
\ENSURE A minimum defensive alliance $S$ or \FALSE
\FOR{$e=uv\in E^+\cup E^- $}
    \STATE\textbf{if} $\deg^-(u)\leq 1$ \textbf{then return} $S=\{u\}$
    \STATE\textbf{else if} $\deg^-(v)\leq 1$ \textbf{then return} $S=\{v\}$    
    \STATE\textbf{else if} $\deg^-(u)=\deg^-(v)=2$ \textbf{then} %\STATE \qquad
    \textbf{return} $S=\{u,v\}$
    %\IF{$\deg^-(u)\leq 1$}
    %\RETURN $S=\{u\}$
    %\ELSIF{$\deg^-(v)\leq 1$}    
    %\RETURN $S=\{v\}$
    %\ELSIF{$\deg^-(u)=\deg^-(v)=2$}
    %\RETURN $S=\{u,v\}$
    %\ENDIF
\ENDFOR
\RETURN \FALSE
\end{algorithmic}          
\end{algorithm}
}
\begin{toappendix}
\begin{proof}
Let $G=(V,E^+,E^-)$ be a signed subcubic graph.
First, we can see if any single vertex forms a defensive alliance on its own. This is the case in particular if $G$ contains a vertex of degree one, but the general condition is stated in \autoref{thm:asd=1}. Similarly, we can check if two adjacent vertices form a defensive alliance, as also specified    in \autoref{thm:asd=1} as a combinatorial condition. Now, two cases remain. If $\delta^-(G)=3$, this means that the graph is 3-regular and all its edges are negative, i.e., $\Delta^+(G)=0$. Then, $G$ does not possess a defensive alliance by \autoref{cor_v}. If  $\delta^-(G)=2$, we know by assuming $a_{sd}(G)>2$ that adjacent vertices $u,v$ never satisfy $\deg^-(v)=\deg^-(u)=2$, but $\min\{\deg^-(u),\deg^-(v)\}\geq 2$. Assume that $a_{sd}(G)>2$ and that $u\in S$ for some hypothetical defensive alliance~$S$. Then, $\deg^-(u)=2$ and (at best) $\deg^+(u)=1$. However, now all neighbors $v$ of $u$ must satisfy $\deg^-(v)=3$, i.e., $v\notin S$, so that we can conclude that such an alliance~$S$ does not exist. 
\end{proof}

\begin{corollary}\label{cor:subcubic-leq2}
If a subcubic graph possesses any defensive alliance~$D$, then its size~$|D|$ is at most~2.
\end{corollary}

Let us mention that later, we are showing that deciding alliability in signed graphs of maximum degree five is already \NP-hard. 
\end{toappendix}
The following structural observation is sometimes helpful:
\begin{propositionrep}\label{prop:min-def-all-connected}
If $S$ is a minimum-size defensive alliance in~$G=(V,E^+,E^-)$, then $S$ is connected in the underlying unsigned graph~$(V,E^+\cup E^-)$. 
\end{propositionrep}
\begin{proof}
%\todo{Proof added because of the comment of rev. 7}
Let $S$ be a minimum-size defensive alliance in~$G=(V,E^+,E^-)$. Let $S'\subseteq S$ be a connected component of the induced graph $G^{\text{uns}}[S]$, where $G^{\text{uns}}=(V,E^+\cup E^-)$ is the underlying unsigned graph of~$G$. Consider some $v\in S'$. As $S'$ is a connected component, within~$G$ we find $\deg_S^+(v)=\deg_{S'}^+(v)$, $\deg_S^-(v)=\deg_{S'}^-(v)$ and  $\deg_{\overline S}^-(v)=\deg_{\overline{S'}}^-(v)$. Hence, because $S$ is a defensive alliance, $S'$ must be also a defensive alliance. As $S$ is a minimum-size defensive alliance, $S=S'$ follows. Therefore, $S$ is connected in~$G^{\text{uns}}$. 
\end{proof}
\begin{toappendix}
As %\todo{remark added for clarification} 
can be seen by the previous proof, we can generalize the statement of \autoref{prop:min-def-all-connected}  by requiring an inclusion-minimal defensive alliance~$S$ instead of a minimum-size one.
\end{toappendix}

Hence, when looking for a DA~$S$ of size at most~$k$, we can assume that the diameter of $(V,E^+\cup E^-)[S]$ is at most~$k-1$.

As mentioned above, balancedness is an important notion in signed networks, in particular in connection with complete graphs. Hence, the next theorem is an \longversion{interesting and }important combinatorial result.  \longversion{This also means that}\shortversion{Hence,} we can determine $a_{sd}(G)$ for any weakly balanced signed complete graph $G$ in polynomial time.

\begin{theorem}\label{thm:complete-def-alliance}
For any signed complete graph $G=(V,E^+,E^-)$, $n=|V|$, with $E^-\neq\emptyset$, \longversion{we can determine its defensive alliance number $a_{sd}(G)$ in the following cases.  \begin{itemize}
    \item [1)] If $G$ is balanced with partition $(V_1,V_2)$, where $|V_1|\geq |V_2|$, then $a_{sd}(G)=|V_2|$. Moreover, any subset $S\subseteq V_1$ with $|S|=|V_2|$ is a minimum defensive alliance.
    \item [2)] I}\shortversion{i}f $G$ is $k$-balanced with partition $(V_1,V_2,\ldots,V_k)$, 
    where $|V_1|\geq |V_2| \geq \ldots \geq |V_k|$, then we find:
    \begin{itemize}
     \item [i)] If $|V_1| \geq \frac{n}{3}$ and $|V_2| \geq \frac{n}{3}$, then $ a_{sd}(G)=2 \left\lceil \frac{n-|V_2|}{2} \right\rceil$. Any \longversion{subset $S_1$ of $V_1$}\shortversion{$S_1\subseteq V_1$} and any \longversion{subset $S_2$ of $V_2$}\shortversion{$S_2\subseteq V_2$} with $|S_1|=|S_2|=\left\lceil \frac{n-|V_2|}{2} \right\rceil$ forms a minimum DA. 
     \item [ii)] If $|V_1| \geq \frac{n}{2}$, then $a_{sd}(G)=n-|V_1|$. Any subset $S$ of $V_1$ with $|S|=n-|V_1|$ is a minimum DA.
     \item [iii)] Otherwise, there is no defensive alliance at all.
    \end{itemize}
\longversion{  \end{itemize}}
\end{theorem}

\begin{proof}
\longversion{As 1) is obviously a special case of 2), we only consider $k$-balanced complete signed graphs in the following. }
    %\begin{itemize}
    %    \item [1)]\todo{By the previous comments, I did not read this part of the proof.} Firstly, note that $\forall S \subseteq V_1$ with $|S|=|V_2|$, $S$ is a defensive alliance in $K_n$. And any subset of $V_2$ can not form a defensive alliance (except when $|V_1|=|V_2|$). Secondly, we will prove that there is no smaller defensive alliances consisting of vertices from two different parts by contradiction. Assume that there exists a defensive alliance $S=A\cup B$, where $A\subseteq V_1, B\subseteq V_2$ are nonempty. Then $|A|+|B|\textless|V_2|$. But for $\forall x\in A$, we have $|A|\geq |V_2|-|B|$ i.e. $|A|+|B|\geq |V_2|$, this leads to a contradiction. 
    %    \item [2)] 
 Slightly abusing notation, we will denote the cluster $V_i$ to which $x\in V$ belongs as $V_x$. Suppose that the nonempty subset $S\subseteq V$ is a defensive alliance, and for all $x\in S$, let $V'_{x}$ denote the subset of~$S$ which includes the vertices that are from the same partition part as~$x$, i.e.,  $V'_x=S\cap V_x$. By definition of a defensive alliance, we have, for any $x \in S$, a)~$\deg_S^+(x)+1=|V'_{x}|\geq \deg_S^-(x)=|S|-|V'_{x}|$ and b)~$\deg_S^+(x)+1=|V'_{x}|\geq \deg_{\overline{S}}^-(x)=n-|S|-|V_x|+|V'_{x}|$. So, for all $x \in S$, $2|V_x|\geq2|V'_{x}|%\geq|S|
 \geq n-|V_x|$, i.e., $|V_x|\geq \frac{n}{3}$. Hence, vertices in~$S$ are from at most three different partition parts. We have to discuss the situation $S=V'_{1}\cup V'_{2} \cup V'_{3}$, where $V'_{1}\subseteq V_1$, $V'_{2}\subseteq V_2$, and $V'_{3}\subseteq V_3$. We further know that $|V_1|\geq |V_2|\geq |V_3|\geq \frac{n}{3}$. \longversion{The following}\shortversion{A detailed} discussion shows the subcases i), ii) and iii)\longversion{ of 2) as in the statement of the theorem} by distinguishing whether there are three, two or one non-empty subset(s) of the form $V'_{x}$.
\begin{toappendix}

\shortversion{Detailed case distinction in the proof of \autoref{thm:complete-def-alliance}:} 
\smallskip
\noindent        
\underline{Case one}: Should we have three non-empty parts $V_1',V_2',V_3'$, then $V_3\neq\emptyset$ and hence  $|V_1|=|V_2|=|V_3|=\frac{n}{3}$. From a) above, $|V'_{1}|\geq |V'_{2}|+|V'_{3}|$,  $|V'_{2}|\geq |V'_{1}|+|V'_{3}|$, and  $|V'_{3}|\geq |V'_{1}|+|V'_{2}|$. So we have $|V'_{1}|=|V'_{2}|=|V'_{3}|=0$, that is to say, there is no such defensive alliance. \shortversion{Refer}\longversion{This corresponds} to part 2)iii). 
  
\smallskip
\noindent        
\underline{Case two}: $S= V'_{1}\cup V'_{2}$ with $V'_2\neq\emptyset$. 
%where $V'_{1}\subseteq V_1, V'_{2}\subseteq V_2,$ and
We know $|V_1|\geq|V_2|\geq \frac{n}{3}$. From a) above, we know $|V'_1|\geq |S|-|V'_1|=|V'_2|$ and $|V'_2|\geq |S|-|V'_2|=|V'_1|$, i.e., 
$|V'_{1}|=|V'_{2}|$.  From b) above,  %$|V'_{1}|\geq n-|V_{1}|-|V'_{2}|$ and  
$|V'_{2}|\geq n-|V_{2}|-|V'_{1}|$. Hence, $|S|=|V'_{1}|+|V'_{2}|\geq n-|V_2|$. On the other  side, one can verify that any subset $V'_{1}$ of $V_1$ and any subset $V'_{2}$ of $V_2$ with $|V'_{1}|=|V'_{2}|=\left\lceil \frac{n-|V_2|}{2} \right\rceil$ forms a defensive alliance\longversion{. This corresponds to}\shortversion{, see} part 2)i). 

\smallskip
\noindent        
\underline{Case three}: $S=V'_{1}\subseteq V_1$. From b)\longversion{ above}, $|V'_{1}|\geq n-|V_{1}|$. Trivially, $|V_1|\geq|V'_{1}|$. Hence, $|V_1|\geq \frac{n}{2}$ and $|S|\geq n-|V_{1}|$. Similarly, we can verify that any subset $V'_{1}$ of $|V_1|$ with $|V'_{1}|=n-|V_1|$ is a defensive alliance\longversion{. This corresponds to}\shortversion{, see} part 2)ii). 
   % \end{itemize}   
\end{toappendix}
\end{proof}
This result also shows the similarities between the notion of a defensive alliance and that of a clustering: A DA is often a cluster, or if not, it is stretching over at most two of them.
\shortversion{In the following, we are formally defining the decision version of the problems that we consider in this paper.}

\section{Formal Problem Definitions} 

\longversion{We are now formally defining the decision version of the problems that we consider in this paper. }As we will show \NP-hardness of the first two problems, most of the following algorithmic research is devoted to finding solutions in special situations, also employing the toolbox of parameterized algorithms. 

\problemdef{Defensive (Pointed) Alliability} {(\textsc{Def(P)All})}{A signed network $G$ (and a vertex~$v$)}{Is there a DA in $G$ (containing~$v$)?}

\problemdef{Minimum Defensive Alliance}{(\textsc{MinDefAll})}{A signed network $G$ and an integer~$k\geq0$}{Is there a DA in $G$ of size at most~$k$?}

In the spirit of \textsc{Correlation Clustering}, we are also discussing edge flips. More formally, let $G=(V,E^+,E^-)$ be a signed graph. 
For a set $T\subseteq E=E^+\cup E^-$, let $G_{T}$ be the signed graph obtained after flipping the edges of~$G$ in~$T$, i.e., $G_T=(V,E^+\mathbin{\triangle}T,E^-\mathbin{\triangle}T)$, where $\mathbin{\triangle}$ denotes the symmetric set difference. Now, \textsc{Correlation Clustering} could be defined as the question to decide, given $G$ and~$k$, if some edge set $T$ exists, $|T|\leq k$, such that $G_T$ is weakly balanced. This question is also known to be \NP-complete, see \cite{BanBluCha2004}. By the mentioned \NP-hardness result of \textsc{Defensive Alliability}, this would be also true for the question to flip at most~$k$ edges to make the graph alliable. However, the following question, whose clustering variant is trivial, is an interesting variation in the context of alliances as we will present a non-trivial algorithm solving it below.

\problemdef{Defensive Alliance Building}{(\textsc{DefAllBuild})}{A signed network $G=(V,E^+,E^-)$, a vertex set $D$ and an integer~$k\geq0$}{Is there a $T\subseteq E=E^+\cup E^-$ of size at most~$k$ such that $D$ is a DA in $G_T$?}

In other words, we ask if one can flip $k$ edge signs in~$G$ so that $D$ becomes a DA. 
This problem can be motivated by thinking of $D$ as an ideal alliance, which basically has to be built by creating friendly relationships out of previously unfriendly ones.  In other words, we might want to turn a specific group of countries or people into a defensive alliance and we want to know how much we have to pay for it.  

\section{Alliability}

It is possible that no defensive alliance exists in a signed network. As an extreme example, consider  any  signed network with only negative edges and minimum degree at least~$3$. So we first consider the existence of defensive alliance in signed networks, \longversion{which is }called \textsc{Defensive (pointed) Alliability}.

%\begin{definition}
%    (\textsc{Defensive Alliability})\\
%    \textbf{Input:} A signed network $G$\\
%    \textbf{Question:} Is there any defensive alliance in $G$?
%\end{definition}

As already discussed above, the question of alliability makes no sense for the unsigned counterpart, while, surprinsingly, we obtain hardness results for the signed version.

\begin{theoremrep}\label{thm:def-alliable}
\textsc{DefPAll} and \textsc{DefAll} are \NP-complete.
\end{theoremrep}

The reduction is based on the \NP-hardness of a well-known variation of the \textsc{Satisfiability} problem, namely \textsc{NAE-3SAT}, see \cite{Sch78}.
\begin{toappendix}
According to Schaeffer, \textsc{NAE-3SAT} can be defined as a coloring problem as follows:
\problemdef{Not-All-Equal 3-Satisfiability}{(\textsc{NAE-3SAT})}{A base set~$X$, a collection of sets $C_1,\dots,C_m\subseteq X$, each having at most three members}{Is there a mapping $A:X\to\{0,1\}$ such that $A(C_i)=\{0,1\}$ for all $1\leq i\leq m$?}
From the coloring point of view, this means that $X$ is colored with two colors such that no set $C_i$ is all one color. From the logical perspective, we can think of $X$ being a set of variables and the coloring being an assignment of the truth values~$1$ (or \texttt{true}) and~$0$ (or \texttt{false}). Notice that in Schaeffer's definition, all clauses are monotone, more precisely, all literals are positive, but clearly the problem will not become easier if negative literals are allowed.

%\todo[inline]{Seeing https://en.wikipedia.org/wiki/Not-all-equal\_3-satisfiability, %https://en.wikipedia.org/wiki/Not-all-equal_3-satisfiability 
%it is clear that there cannot be subexponential algorithms for \textsc{NAE-3SAT} under ETH, so we also get an ETH-result by our reduction. It would be good to fill in the details here!}
\end{toappendix}

\begin{proof}
For a signed network $G=(V,E^+,E^-)$, obviously, \textsc{Defensive (Pointed) Alliability} is in \NP, by checking the two defensive conditions of every vertex in a solution $S$, which runs in polynomial time $\mathcal{O}(|S|^2(|E^+|+|E^-|))$.

We complete our proof, showing \NP-hardness by reducing from \textsc{NAE-3SAT}, formulated in its logic version. Let $\phi=C_1\wedge\dots\wedge C_m$ be a boolean formula with variable set $X$ (with $|X|=n$) in which each clause~$C_i$ includes exactly $3$ positive literals. The question is whether there exists a satisfying assignment such that in each clause there is a literal assigned to \texttt{false}. We are going to describe a signed graph $G=(V,E^+,E^-)$, where $V$ may contain a special vertex~$v$ in the pointed variation of our problem. We first define some auxiliary clique structures. Define $NC_v=(V_v,\emptyset, E^-_v)$, which is is a negative clique with $m+2$ many vertices, containing the possibly special vertex~$v$, as well as $NC_{x,i}=(V_{x,i},\emptyset,E^-_{x,i})$, with $x\in X$ and $i\in \mathbb{N^+}$, which is a negative clique with four vertices; both types of negative cliques serve as not-in-the-solution gadgets, differentiated as \emph{big} and \emph{small}, respectively. More technically, we can put $V_{x,i}\coloneqq \{x_{i,1},x_{i,2},x_{i,3},x_{i,4}\}$.
For each variable $x$ in~$X$, we define 
\begin{itemize}
    \item $C(x)=\big\{C_{j_1},\dots,C_{j_{n_x}}\mid x\in C_{j_i},i\in\{1,\dots,n_x\}, j_i\in\{1,\dots,m\}\big\}$, i.e., \item $n_x$ is the number of clauses that contain the variable~$x$, 
    and
    \item  $X'(x)\coloneqq\big\{x_h\mid h\in \{1,\dots,2n_x\}\big\} $.
\end{itemize}  
Then we construct the reduction $R(\langle \phi\rangle):=$ 
\begin{itemize}
    \item [1)] Build a signed graph $G=(V,E^+,E^-)$ as follows:\\
    \begin{equation*}
        \begin{split}
            V={}& V_v\cup \left\{c_j\mid j\in \{1,\dots,m\}\right\}\cup\bigcup_{x\in X} \left(X'(x)\cup \bigcup_{l\in \{1,\dots,4n_x\}} V_{x,l}\right)\\
            E^+={}&\{vc_j\mid j\in\{1,\dots ,m\}\}\cup\big\{x_{2p-1}x_{2p} \mid p\in \{1,\dots,n_x\}, x\in X\big\}\\
            E^-={}&E^-_v\cup\left(\bigcup_{x\in X}\bigcup_{l\in \{1,\dots,4n_x\}}E_{x,l}^-\right)\cup \left(\bigcup_{x\in X}\big\{c_{j_p}x_{2p}\mid p\in \{1,\dots,n_x\},C_{j_p}\in C(x)\big\}\right)\cup{}
            \\&\big\{x_{2p}x_{2p+1}, x_{2n_x}x_1\mid x\in X, p\in \{1,\dots,n_x-1\} \big\} \cup{}
            \\& \big\{ x_{2p-1}x_{4p-3,1},x_{2p-1}x_{4p-2,1}, x_{2p}x_{4p-1,1},x_{2p}x_{4p,1}  \mid x\in X,  p\in \{1,\dots,n_x\}\big\}
        \end{split}
    \end{equation*}
\item[2)] Return $\left \langle G, v \right \rangle$ or just $G$ if we talk about \textsc{DefAll}.      
\end{itemize}

 \begin{figure}[bt]
    \centering
    
\begin{subfigure}[t]{1.\textwidth}
    \centering
    \begin{tikzpicture}
        \tikzset{every node/.style={fill = white,circle,minimum size=0.3cm}}
        
        \node[draw] (x1) at (-0.9,3) {};
        \node[draw] (x2) at (0.5,3) {};
        \node[draw] (x3) at (-2.1,2) {};
        \node[draw] (x4) at (-3,0.8) {};
        \node[draw] (x6) at (-2.1,-2) {};
        \node[draw] (x5) at (-3,-0.8) {};
        \node[draw] (x7) at (-0.9,-3) {};

        \node[draw] (x8) at (2.1,2) {};
        \node[draw] (x9) at (3,0.8) {};
        \node[draw] (x11) at (2.1,-2) {};
        \node[draw] (x10) at (3,-0.8) {};
        \node[draw] (x12) at (0.9,-3) {};

        \node[] (y1) at (-0.9,4) {};
        \node[] (y2) at (-4,1.3) {};
        \node[] (y3) at (-2.9,-2.8) {};
        \node[] (y4) at (3.1,3) {};
        \node[] (y5) at (4,-1.5) {};
        \node[] (y6) at (1.6,-4) {};

        \node[draw,rectangle] (z11) at (-1.25,2.1) {};
        \node[draw,rectangle] (z12) at (-0.45,2.1) {};
        \node[draw,rectangle] (z21) at (0.1,3.9) {};
        \node[draw,rectangle] (z22) at (0.9,3.9) {};
        \node[draw,rectangle] (z31) at (-2.5,2.9) {};
        \node[draw,rectangle] (z32) at (-3,2.3) {};
        \node[draw,rectangle] (z41) at (-2,1.1) {};
        \node[draw,rectangle] (z42) at (-2,0.3) {};\
        \node[draw,rectangle] (z51) at (-4,-0.4) {};
        \node[draw,rectangle] (z52) at (-4,-1.2) {};
        \node[draw,rectangle] (z61) at (-1.9,-1) {};
        \node[draw,rectangle] (z62) at (-1.2,-1.5) {};
        \node[draw,rectangle] (z71) at (-1.3,-4) {};
        \node[draw,rectangle] (z72) at (-0.5,-4) {};

        \node[draw,rectangle] (z81) at (1.2,1.6) {};
        \node[draw,rectangle] (z82) at (1.6,1.1) {};
        \node[draw,rectangle] (z91) at (4,1.2) {};
        \node[draw,rectangle] (z92) at (4,0.5) {};\
        \node[draw,rectangle] (z101) at (2.2,0) {};
        \node[draw,rectangle] (z102) at (1.9,-0.6) {};
        \node[draw,rectangle] (z111) at (2.3,-3) {};
        \node[draw,rectangle] (z112) at (3,-2.5) {};
        \node[draw,rectangle] (z121) at (0.7,-2) {};
        \node[draw,rectangle] (z122) at (0.2,-2.3) {};

        \path (x1) edge[-] (x2);
        \path (x3) edge[-] (x4);
        \path (x5) edge[-] (x6);
        \path (x1) edge[dashed] (x3);
        \path (x4) edge[dashed] (x5);
        \path (x6) edge[dashed] (x7);
        \path (x8) edge[-] (x9);
        \path (x10) edge[-] (x11);
        \path (x8) edge[dashed] (x2);
        \path (x9) edge[dashed] (x10);
        \path (x11) edge[dashed] (x12);
        \path (x7) -- node[auto=false]{\ldots} (x12);

        \path (x1) edge[dashed] (y1);
        \path (x4) edge[dashed] (y2);
        \path (x6) edge[dashed] (y3);
        \path (x8) edge[dashed] (y4);
        \path (x10) edge[dashed] (y5);
        \path (x12) edge[dashed] (y6);

        \path (x1) edge[dashed] (z11);
        \path (x1) edge[dashed] (z12);
        \path (x2) edge[dashed] (z21);
        \path (x2) edge[dashed] (z22);
        \path (x3) edge[dashed] (z31);
        \path (x3) edge[dashed] (z32);
        \path (x4) edge[dashed] (z41);
        \path (x4) edge[dashed] (z42);
        \path (x5) edge[dashed] (z51);
        \path (x5) edge[dashed] (z52);
        \path (x6) edge[dashed] (z61);
        \path (x6) edge[dashed] (z62);
        \path (x7) edge[dashed] (z71);
        \path (x7) edge[dashed] (z72);
        \path (x8) edge[dashed] (z81);
        \path (x8) edge[dashed] (z82);
        \path (x9) edge[dashed] (z91);
        \path (x9) edge[dashed] (z92);
        \path (x10) edge[dashed] (z101);
        \path (x10) edge[dashed] (z102);
        \path (x11) edge[dashed] (z111);
        \path (x11) edge[dashed] (z112);
        \path (x12) edge[dashed] (z121);
        \path (x12) edge[dashed] (z122);
    \end{tikzpicture}
    \subcaption{Variable gadget: the squares represent small not-in-the-solution gadgets.}
    \label{fig:AlliabilityVariableGadget}
\end{subfigure}
\begin{subfigure}[t]{.45\textwidth}
   \centering  
   \begin{tikzpicture}[rotate=90]
        \tikzset{node/.style={fill = white,circle,minimum size=0.3cm}}

        \node[draw,circle] (c) at (0,0) {};
        \node[draw,diamond] (v) at (-1,0) {};
        \node[] (x1) at (1,1) {};
        \node[] (x2) at (1,0) {};
        \node[] (x3) at (1,-1) {};

        \path (v) edge[-] (c);
        \path (c) edge[dashed] (x1);
        \path (c) edge[dashed] (x2);
        \path (c) edge[dashed] (x3);
   \end{tikzpicture}
   \subcaption{Clause gadget: the diamond represents vertex $v$.}
   \label{fig:AlliabilityClauseGadget}
\end{subfigure}
\begin{subfigure}[t]{.44\textwidth}
   \centering  
   \begin{tikzpicture}[rotate=270]
        \tikzset{node/.style={fill = white,circle,minimum size=0.3cm}}
        \node[draw,circle] (x1) at (0,0) {};
        \node[draw,circle] (x2) at (-1,-1) {};
        \node[draw,circle] (x3) at (-1,1) {};
        \node[draw,circle] (x4) at (-2,0) {};
        \node[] (x5) at (-3,0) {};      
        
        \path (x1) edge[dashed] (x2);
        \path (x1) edge[dashed] (x3);
        \path (x1) edge[dashed] (x4);
        \path (x2) edge[dashed] (x3);
        \path (x2) edge[dashed] (x4);
        \path (x3) edge[dashed] (x4);
        \path (x4) edge[dashed] (x5);      
   \end{tikzpicture}
   \subcaption{A small not-in-the-solution gadget.}
    \label{fig:AlliabilityNSGadget}
\end{subfigure}    
    \caption{Reduction construction for \autoref{thm:def-alliable}. }
    \label{fig:alliability}
\end{figure}
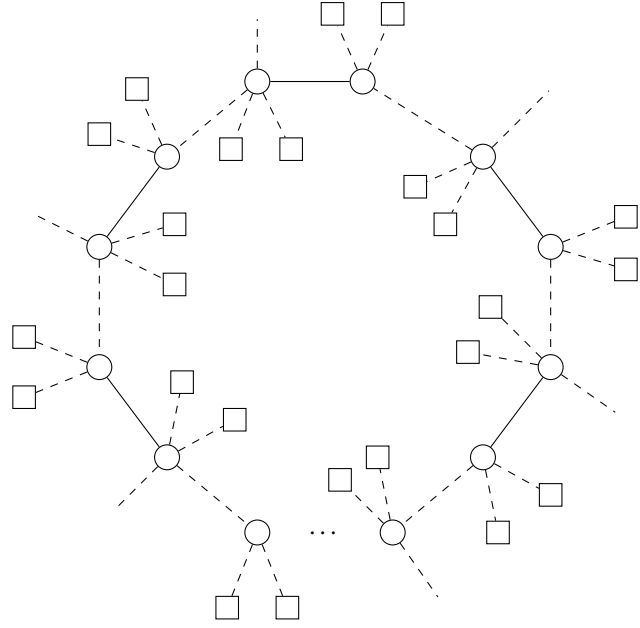
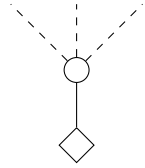
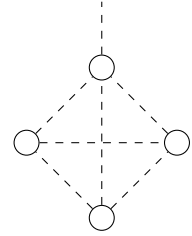

If $\phi=C_1\wedge\dots\wedge C_m$ is a \yes-instance of \textsc{NAE-3SAT}, where $C_i=x_{i1}\vee x_{i2}\vee x_{i3}$, $x_{ij}\in X$ and $i\in \{1,\dots,m\}$, $j\in \{1,2,3\}$, then there exists a satisfying assignment~$A$ for the variables $X$ such that, for each $C_i$, there are $j_1,j_2 \in \{1,2,3\}$ such that $x_{ij_1}=1$ and $x_{ij_2}=0$. Let $S=\{x\in X\mid A(x)=0\}$. We use the assignment~$A$ to show that $G$ possesses a defensive alliance~$D$ containing the vertex $v$: \[D\coloneqq\{v\}\cup\big\{c_i\mid i\in \{1,\dots,m\}\big\}\cup \bigcup_{x\in S} X'(x)\] We have to check the two defensive alliance conditions for each vertex in this set~$D$.
For the vertex $v$ fixed to be in~$D$, 
\begin{enumerate}
 \item $\deg^+_D(v)+1=m+1>0=\deg^-_D(v)$, and
 \item $\deg^+_D(v)+1=m+1=\deg^-_{\overline{D}}(v)$.
\end{enumerate}  
Consider any~$c_i$ for $1\leq i\leq m$.
Observe that $\deg^-_D(c_i)=\deg_{S_X}^-(c_i)$, with $S_X=\bigcup_{x\in S} X'(x)$, so that $\deg^-_D(c_i)\in\{1,2\}$, as otherwise the NAE-assignment~$A$ that determines~$S$ would set all variables in $C_i$ either to~0 or to~1, a contradiction. Similarly, $\deg^-_{\overline{D}}(c_i)=\deg^-_{S_X'}(c_i)$, with $S_X'=\bigcup_{x\in X\setminus S} X'(x)$, so that again $\deg^-_{\overline{D}}\in\{1,2\}$.
Hence,  1) $\deg^+_D(c_i)+1=2\geq\deg^-_D({c_i})$; 2) $\deg^+_D({c_i})+1=2\geq\deg^-_{\overline{D}}({c_i})$. 

For each vertex $x_h\in D$ associated to $x\in X$, 1) $\deg^+_D(x_h)+1=2\geq\deg^-_D(x_h)$; 2) $\deg^+_D(x_h)+1=2=\deg^-_{\overline{D}}(x_h)$. 

Therefore, $D$ is a defensive alliance which contains the vertex~$v$.

Let $D$ be a defensive alliance containing vertex~$v$ of the signed graph obtained from $R(\left\langle \phi\right\rangle)$. First observe that the not-in-the-solution gadgets deserve their name, i.e., $D\cap \left(V_v\cup \bigcup_{x\in X,l\in \{1,\dots,4n_x\}}V_{x,l}\right)=\emptyset$.
In order to defend vertex~$v$, all of its positive neighbors must be in~$D$, i.e., $\big\{c_j\mid j\in \{1,\dots,m\}\big\}\subset D\,.$ For each vertex $c_i\in D$, at least one of its negative neighbors should be in~$D$, and at least one of its negative neighbors should not be in~$D$, because $c_i$ has exactly one positive neighbor, which is~$v$. Then assume that there is some $x_{2p}\in D$. As its two negative neighbors $x_{4p-1,1},x_{4p,1}$ are in small not-in-the-solution gadgets, $x_{4p-1,1},x_{4p,1}\notin D$, we know that the other two negative neighbors $x_{2p-1},x_{2p+1}$ must belong to $D$, as only (possibly) the positive neighbor $c_{j_p}$ can defend $x_{2p}$; but it also must defend it, so that $c_{j_p}\in D$ is enforced. The case of $x_{2p+1}\in D$ is very similar (including the boundary case of $x_1$). Now, the vertex has degree four and negative degree three. As two of its negative neighbors belong to small not-in-the-solution gadgets, both the negative and the positive remaining neighbors must belong to~$D$, as well, bringing us back to the previous case. Hence, we see that $x_h\in D$ for some $h\in\{1,\dots,2n_x\}$ \iffl $x_h\in D$ for all $h\in\{1,\dots,2n_x\}$.  In other words, for each $x\in X$, $X'(x)\subseteq D$ or $X'(x)\cap D=\emptyset$.
If $X'(x)\subseteq D$, then for the set of clause vertices $c(x)$ that correspond to the set of clauses $C(x)$, $c(x)\subseteq D$.
%Following an induction\todo{is this induction? --- yes; HF: I still doubt it}, $X'(x)\subseteq D, C(x)\subseteq D$.\todo{Also, $C(x)$ is a set of clauses} 
Then we can obtain a satisfying assignment~$A$ for~$\phi$ by assigning the value~1, i.e.,  \texttt{true}, to $x\in X$ if $X'(x)\subseteq D$ or the value~0, i.e., \texttt{false}, to $x$ if $X'(x)\cap D=\emptyset$. Therefore, $A$ is a satisfying assignment for~$\phi$, whose values in each clause are not all equal to each other. 
    
Furthermore, without the fixed vertex $v$, if $c_i \in D$, then $v\in D$, because $v$ is the only positive neighbor of each $c_i$. If some $x_{2p}\in D$, according to the above, we know that the corresponding set $c(x)$ of clause vertices  satisfies $c(x)\subseteq D$, so that  $v\in D$.

It is easy to see that the described procedure for building $R(\langle \phi\rangle)$ from~$\phi$ can be executed in polynomial time. Consequently, both \textsc{DefPAll} and \textsc{DefAll} are \NP-complete.
\end{proof}

\begin{toappendix}
\longversion{\noindent}By analyzing the previous proof again more carefully, and taking the notation from that proof, we find:

\begin{lemma}\label{lem:all-con}
The constructed signed graph $G$ has at most $56m+1$ many vertices.
\end{lemma}

\begin{proof}Let us analyze the parts of the vertex set $V$ of $G$ one-by-one.
\begin{itemize}
    \item $|V_v|=m+2$ by definition.
    \item $|\{c_j\mid j\in \{1,\dots,m\}\}|=m$ by definition.
    \item $\left|\bigcup_{x\in X}X'(x)\right|=\sum_{x\in X}2n_x \leq  6m$ by double-counting, as each clause contains at most 3 literals and $n_x$ is the number of clauses in which variable $x$ occurs.
    \item Similarly,  $\left|\bigcup_{x\in X}\bigcup_{l\in \{1,\dots,4n_x\}}V_{x,l}\right|=\sum_{x\in X}16n_x=48m$.
\end{itemize}
Altogether, $56m+2=(m+2)+m+6m+48m$ as claimed.
\end{proof}
%\todo[inline]{We should make the whole chain of ETH results explicit in the appendix. I try to make a start now.}

\begin{theorem}
Given an instance $I=(X,\mathcal{C})$ of \textsc{3SAT}, with $n=|X|$ variables and $m=|\mathcal{C}|$ clauses, one can construct in polynomial time an instance  $I'=(X',\mathcal{C}')$ of \textsc{NAE-3SAT}, with $n'=|X'|$ and $m'=|\mathcal{C}'|$ such that $I$ is satisfiable \iffl $I'$ has a solution and such that $n'=\mathcal{O}(n+m)$ and $m'=\mathcal{O}(n+m)$.
\end{theorem}

\begin{proof}
The claimed reduction is similar to the one indicated in https://en.wikipedia.org/wiki/Not-all-equal\_3-satisfiability.
For the given \textsc{3SAT}-instance  $I=(X,\mathcal{C})$, we can assume that all clauses in $\mathcal{C}$ contain three literals and that no variable occurs twice in any clause. 
Let $L(X)=\{x,\neg x\mid x\in X\}$ be the set of literals of~$I$. Further, assume that we have arbitrarily fixed an ordering on the literals in $L(C)$ for each $C\in\mathcal{C}$, defining three mappings $L_1,L_2,L_3:\mathcal{C}\to L(X)$.  
We construct $I'=(X',\mathcal{C}')$ as follows:
\begin{itemize}
    \item $X'\coloneqq X\cup \{\hat x\mid x\in X\}\cup\{s\}\cup \{x_C, {\hat x}_C\mid C\in \mathcal{C}\}$.\\
    Define the injection $f:L(x)\to X'$ by $x\mapsto x$ and $\neg x\mapsto \hat x$ for each $x\in X$.
    \item $\mathcal{C}'\coloneqq \big\{\{f(L_1(C)),f(L_2(C)),x_C\},\{{\hat x}_C,f(L_3(C)),s\},\{x_C,{\hat x}_C\}\mid C\in\mathcal{C}\big\}\cup \big\{\{x,\hat x\}\mid x\in X\big\}$.
\end{itemize}
First, assume that $A:X\to \{0,1\}$ is a satisfying assignment for~$I$. Then $A$ is extended to $A':X'\to\{0,1\}$ by letting $A'(x)=A(x)$ for $x\in X$, $A'(\hat x)=\neg A(x)$ (flipping $0$ to $1$ and $1$ to $0$), while $A'(x_C)=1$ and $A'({\hat x}_C)=0$ \iffl $A(L_1(C))=A(L_2(C))=0$ for all $C\in\mathcal{C}$ and $A'(s)=1$. Observe that for each $C'\in\mathcal{C'}$, one of its variables is set to~0 and another one is set to~1.

Conversely, if $A':X'\to \{0,1\}$ is a not-all-equal assignment of $I'$, then observe that the assignment $A'':X'\to \{0,1\}$ that flips all variables in comparison with $A'$, i.e., $A''(x)=\neg(A'(x))$, is also a not-all-equal assignment of $I'$. Hence, we can assume, w.l.o.g., that $A'(s)=0$. Let $A:X\to\{0,1\}$ be simply the restriction of $A'$ to $X$.
Our setting guarantees that, for each $c\in\mathcal{C}$, by observing the clause $\{{\hat x}_C,f(L_3(C)),s\}$, either (1) $A'({\hat x}_C)=1$ and hence $A'(x_C)=0$ due to the clause $\{x_C,{\hat x}_C\}$, or (2) 
$A'(f(L_3(C)))=1$, so that $C$ is satisfied immediately. In case (1), by observing the clause $\{f(L_1(C)),f(L_2(C)),x_C\}$, either $A'(f(L_1(C)))=1$ or  $A'(f(L_2(C)))=1$, i.e., the described assignment for~$I$ indeed satisfies each $C\in\mathcal{C}$.
\end{proof}

\begin{corollary}
Unless ETH breaks, there is no algorithm solving monotone \textsc{NAE-3SAT} instances with $n$ variables and $m$ clauses in time $\mathcal{O}\big(2^{o(n+m)}\big)$.
\end{corollary}

As a further excursion into interesting variants of \textsc{NAE-SAT}, let us make some more observations.
Notice that we could also change the construction slightly by replacing all clauses $\{x,\hat x\}$ of size~2 by the following gadget consisting of 6 clauses (compare \cite{DarDoc2020}):
\[\textbf{NE}(x,\hat x)=\big\{x,\hat x,y\},\{x,\hat x,z\},\{v,y,z\},\{w,y,z\},\{u,v,w\}\big\}\,.\]
In this way, we get an equivalent \textsc{NAE-3SAT} formula, again with $\mathcal{O}(n+m)$ many variables and clauses. Even more, both constructions shown in \cite{DarDoc2020} prove that also for the restriction \textsc{NAE-3SAT-E4} where each clause is monotone and contains three variables and every variable occurs exactly four times, we finally obtain $\mathcal{O}(n+m)$ many variables and clauses, so that we can conclude the following stronger ETH result.

\begin{corollary}
Unless ETH breaks, there is no algorithm solving monotone \textsc{NAE-3SAT-E4} instances with~$n$ variables and $m$ clauses in time $\mathcal{O}\big(2^{o(n+m)}\big)$.
\end{corollary}

Let us now return to our main theme, which is defensive alliances.
\end{toappendix}

Our considerations also allow us to deduce some impossibility results based on the famous Exponential-Time Hypthesis (ETH), see \cite{ImpPatZan2001}.

\begin{corollary}
Assuming ETH, there is no $\mathcal{O}(2^{o(n)})$-time algorithm for solving \textsc{Def(P)All}-instances with $n$ vertices.
\end{corollary}

We can complement our finding by exhibiting polynomial-time algorithms for  \textsc{Def(P)All} on some graph classes.
Above, we \longversion{already }looked at subcubic and weakly balanced signed graphs. %in the previous section and 
We will look into structural parameters like \shortversion{\snd{}}\longversion{bounded signed neighborhood diversity or treewidth, combined with bounded degree,} later.

\section{Building Defensive Alliances}

We turn to the task of building an alliance by flipping edges. 

\begin{theorem}
\label{thm:def-alliance-building}
\textsc{Defensive Alliance Building} can be  solved in polynomial time (details below).    
\end{theorem}

In our algorithm, we will make use of the problem %\todo{slightly changed, as the algorithm becomes simpler} 
\textsc{Upper Degree-Constrained Subgraph} (for short: \textsc{UDCS}). An instance of \textsc{UDCS} is given by an unsigned graph $G=(V,E)$ and an upper bound $u_v\in \mathbb{N}$ for each $v\in V$. The task is to find a subgraph~$H$ with maximum number of edges, such that for each vertex $v\in V$, $\deg_{H}(v) \leq u_v$. This problem was introduced by \cite{Gab83} where the author showed that this problem is solvable in time $\mathcal{O}\left(\sqrt{\sum_{v\in V}u_v}\cdot  \vert E \vert \right)$.

We now describe our algorithm, working on an instance $(G,D,k)$ of \textsc{Defensive Alliance Building}, i.e., $G=(V,E^+,E^-)$ is a signed graph, $D\subseteq V$ and $k\geq 0$.
First, we apply the following \underline{reduction rule}:
    If there exists a $v\in D$ with $z_v\coloneqq\deg_{G,\overline{D}}^-(v)-\deg_{G,D}^+(v)-\deg_{G,D}^-(v) \geq 0$, then set $S=\{ vx \mid x \in N^-_D(v) \}$ and add $z_v$ edges from $\{ vx \mid x \in N^-_{\overline{D}}(v)\}$ to~$S$, which is part of the solution that we build in the sense of the following correctness assertion, where %\todo{added} 
    $G_S=(V,E^+\mathbin{\triangle} S, E^- \mathbin{\triangle} S)$:
\begin{lemmarep}
$\langle G,D,k \rangle $ is a \yes-instance of \longversion{\textsc{Defensive Alliance Building}}\shortversion{\textsc{DefAllBuild}} \iffl $\langle G_S,D,k'\rangle$ is one, with \longversion{$k'=k-z_v-\deg_{G,D}^-(v)=k-(\deg_{G,\overline{D}}^-(v)-\deg_{G,D}^+)$}\shortversion{$k'=k-z_v-\deg_{G,D}^-(v)$}.
\end{lemmarep}

\begin{proof}
This lemma can be shown by induction by using the following claim:
\begin{claim}\label{claim:flip_inside}
        Let $v\in D$ and $S \subseteq E$ be a solution of the \textsc{Defensive Alliance Building} instance $(G,D,k)$ such that there exist edges $vx \in E^-\setminus S$ and $vy \in E\cap S$ with $x\in D$ and $y \notin D$. Then $S'\coloneqq \left(S \cup \{ vx\} \right) \setminus \{vy\} $ is also a solution.
    \end{claim}
    \begin{pfclaim}
        Let $u\in D\setminus \{v\}$. So, $\deg_{G_S,D}^+(u) \leq \deg_{G_{S'},D}^+(u)$, $\deg_{G_S,D}^-(u) \geq \deg_{G_{S'},D}^-(u)$ (with equality if $u\neq x$) and $\deg_{G_S,\overline{D}}^-(u) = \deg_{G_{S'},\overline{D}}^-(u)$. Since $\deg_{G_{S'}D}^+(v)=\deg_{G_{S},D}^+(v) + 1$ and $\deg_{G_{S'},\overline{D}}^-(v)=\deg_{G_S,\overline{D}}^-(v)+1$, as well as $\deg_{G_{S'},D}^-(v) = \deg_{G_{S},D}^-(v)-1$, $D$ is a defensive alliance of~$G_{S'}$, just as it was of~$G_S$.
    \end{pfclaim}
    \noindent
    Applying this claim $z_v$ times proves the lemma.
\end{proof}

From now on, each $v\in D$ fulfills \[\deg_{G,D}^+(v)+\deg_{G,D}^-(v)>\deg_{G,\overline{D}}^-(v)\,.\] Hence, a solution includes only negative edges in the induced signed graph $G[D]$. 
    For the algorithm, we define $ B\coloneqq \{v \in D \mid \deg_{G,D}^+(v) + 1 < \max \{\deg_{G,D}^-(v), \deg_{G,\overline{D}}^-(v))\},$
    collecting those vertices from $D$ that violate DA conditions. In order to quantify this violation, set $b(v)\coloneqq0$ if $v\notin B$ and $b(v)\coloneqq\max\{b_1(v),b_2(v)\}$ if $v\in B$, where 
\begin{align*}
b_1(v)&{}\coloneqq \deg_{G,\overline{D}}^-(v) - \deg_{G,D}^+(v) - 1,\\
b_2(v)&{}\coloneqq \left\lceil \frac{\deg_{G,D}^-(v)- \deg_{G,D}^+(v) - 1}{2}\right\rceil\,.
\end{align*} 
    We run the polynomial-time algorithm for \textsc{UDCS} on the unsigned graph $G^-\coloneqq (V,E^-)$ with $u_v\coloneqq b(v)$. Let $H'=(V,S')$ be the subgraph which is returned by this algorithm. Clearly, $S'$ only includes edges between vertices in~$B$ and, for each $v\in B$ with $\deg_{H'}(v) < b(v)$ and each $u \in N_{G',B}(v)=N_{G,B}^-(v)$,
    $\deg_{H'}(u)=b(u)$. Otherwise, $H'$ would not be maximal. Now, let $S\coloneqq S'$ and $H\coloneqq H'$; for $v\in B$ with $\deg_{H}(v) < b(v)$, add   edges $vx\in E^-\setminus S$ with $x \in D$ to $H$ and $S$ until equality holds. This implies $ \vert S \setminus S'\vert =  \vert S \vert - \vert S'\vert = \sum_{v\in B} b(v)- \deg_{\widetilde{H}}(v)$.

\begin{lemmarep}
$S$ has minimum cardinality \longversion{such that}\shortversion{s.t.} $D$ is a \longversion{defensive alliance}\shortversion{DA} of\shortversion{~$G_S$}\longversion{ $G_S=(V,E^+\mathbin{\triangle} S, E^- \mathbin{\triangle} S)$}.
\end{lemmarep}
\begin{proof}

        First of all, we show that $D$ is a defensive alliance of~$G_S$. Clearly, $S\subseteq E^-$. Therefore, for each $v\in D \setminus B$, 
    \begin{equation*}
        \begin{split}\deg_{G_S,D}^+(v) + 1 &\geq \deg_{G,D}^+(v) + 1 \geq  \max \{\deg_{G,D}^-(v), \deg_{G,\overline{D}}^-(v)\}\\ &\geq \max \{\deg_{G_S,D}^-(v), \deg_{G_S,\overline{D}}^-(v)\}\,.        \end{split}
    \end{equation*}
        Let $v\in B$. Since $S$ includes at least $b(v)$  edges incident to $v$  with two endpoints in~$D$, $\deg_{G_S,\overline{D}}^-(v)=\deg_{G,\overline{D}}^-(v) \leq  \deg_{G,D}^+(v)+ b(v)+1 \leq \deg_{G_S,D}^+(v)+1$ and $\deg_{G_S,\overline{D}}^-(v) \leq \deg_{G,\overline{D}}^-(v) - b(v) \leq  \deg_{G,D}^+(v)+ b(v) +1 \leq \deg_{G_S,D}^+(v)+1$. Thus, $D$ is a defensive alliance. 

Let $T \subseteq E$ be of  minimum cardinality such that $D$ is a defensive alliance of~$G_T$.  Assume there exists a $v\in B$ with $\vert \{vx\in E \mid vx\in T\}\vert < b(v)$. We will show that this assumption leads to a contradiction to the fact that $D$ is a defensive alliance of~$G_T$.

\noindent\underline{Case 1:} 
If $b(v)= \deg_{G,\overline{D}}^-(v) - \deg_{G,D}^+(v) - 1$, then \[\deg_{G_T,\overline{D}}^-(v) = \deg_{G,\overline{D}}^-(v) = \deg_{G,D}^+(v) + 1 + b(v) > \deg_{G_T,D}^+(v) + 1\,.\]

\noindent\underline{Case 2:}
For $b(v)=  \frac{\deg_{G,D}^-(v)- \deg_{G,D}^+(v) - 1}{2}$, we have \[\deg_{G_T,D}^-(v) > \deg_{G,D}^-(v) - b(v) = b(v) + \deg_{G,D}^+(v) +1 > \deg_{G_T,D}^+(v) + 1\,.\] The case $b(v)=  \frac{\deg_{G,D}^-(v)- \deg_{G,D}^+(v) - 1}{2}$ implies \[\deg_{G_T,D}^-(v) > \deg_{G,D}^-(v) - b(v) = b(v) + \deg_{G,D}^+(v) \geq  \deg_{G_T,D}^+(v) + 1\,.\] As each case contradicts the fact that $D$ is a defensive alliance on $G_T$, we conclude that  $\vert \{vx\in E \mid vx\in T\}\vert \geq b(v)$ for each $v\in B$. 

        Let $\widetilde{H} \coloneqq \left( V,\widetilde{S} \right)$ be the output of the \textsc{DCS} algorithm on $\widetilde{G}=(V,T)$ with the bounds from above. Since $\widetilde{G}$ is a subgraph of $G$, $\vert \widetilde{S} \vert \leq \vert S' \vert$. Since $\widetilde{S}$ is maximal, adding an edge to $\widetilde{S}$ would increase $\deg_{\widetilde{H},V}(v)$ for at most on $v\in B$ with $b(v) \geq \deg_{\widetilde{H},V}(v)$.  Therefore, we need to add at least $\sum_{v\in B}\left(b(v)-\deg_{\widetilde{H}}(v)\right)$ many edges to $\widetilde{S}$ to obtain~$T$. Moreover, 
        \[\vert T \vert -\vert \widetilde{S}\vert = \vert T\setminus \widetilde{S} \vert \geq  \sum_{v\in B}\left(b(v)-\deg_{\widetilde{H}}(v)\right) 
                = \left(\sum_{v\in B} b(v) - \deg_{H'}(v) \right) +  \deg_{H'}(v) - \sum_{v\in B} \deg_{\widetilde{H}}(v).
            \]
            By the sentence before the lemma and the handshaking lemma, we observe \[\vert T \vert -\vert \widetilde{S}\vert \geq \vert S\vert -\vert S'\vert + 2\cdot\left(\vert S'\vert + \vert \widetilde{S} \vert \right) = \vert S\vert + \left(\vert S'\vert - \vert \widetilde{S}\vert \right)-\vert \widetilde{S}\vert= \vert S\vert -\vert \widetilde{S}\vert\,\] Thus, $\vert S\vert \leq \vert T\vert$.
\end{proof}

We still have to analyze the running time. The exhaustive employment of the reduction rule runs in time $\mathcal{O}\left( \vert V \vert^2 \right)$. The condition $b(v) \leq \deg^-_{G}(v)\leq \vert V\vert $ can be calculated in linear time for all $v\in V$. Furthermore, the \textsc{DCS}-algorithm runs in time $\mathcal{O}\left(\sqrt{\vert V\vert} \cdot \vert E \vert \right)$. Adding the vertices in the end runs in linear time with respect to the number of edges. Therefore, overall the algorithm runs in time $\mathcal{O}\left( \vert V \vert^2 + \sqrt{\vert V\vert} \cdot \vert E \vert \right)$. 

\section{Minimum Defensive Alliances}

We now describe a formal link between \textsc{Minimum Defensive Alliance} in unsigned and signed graphs; we describe this in the form of a reduction from \longversion{\textsc{Minimum Defensive Alliance} on}\shortversion{the case of} unsigned graphs. This allows us to transfer hardness results known for unsigned graphs to the case of signed graphs. 

\begin{theoremrep}\label{prop:From-unsigned-to-signed}
   There is a polynomial-time transformation that, given an unsigned graph~$G$, produces a signed graph~$G'$ such that a vertex set~$A$ is a defensive alliance in~$G$ \iffl it is a defensive alliance in~$G'$. 
\end{theoremrep}

\begin{proof}
    Let $G=(V,E)$ be an unsigned graph and $k\in\mathbb{N}$. For $v\in V$, define $d'(v) \coloneqq \left\lceil \frac{\deg(v)+1}{2}\right\rceil$, as well as a new set of vertices $M_{v} \coloneqq \{ v_{i,j}\mid i\in \{1,\ldots,d'(v)\}, j\in \{1,2,3,4\}\,\}$. We construct a signed graph $G'=(V',E'^+,E'^-)$ with 
\begin{equation*}
    \begin{split}
        V' \coloneqq{}& V \cup \bigcup_{v\in V} M_v, \quad E'^+ \coloneqq E, \text{ and}\\
        E'^-\coloneqq{}& \big\{\{v,v_{i,1}\}, \{v_{i,j},v_{i,\ell}\}\mid v\in V, i\in \{ 1,\ldots,d'(v)\},\, j,\ell\in \{1,2,3,4\}\big\}. 
    \end{split}
\end{equation*}

In other words, the only positive edges of~$G'$ are inherited from~$G$, and the unsigned graph $(V',E^{\prime -})$ is subgraph of a split graph, with $V\subseteq V'$ forming an independent set and $V'\setminus V$ forming a collection of cliques. 
Since $\deg^-(v_{i,j})\geq 3$ and $\deg^+(v_{i,j})=0$ for all $i\in \{1,\ldots,d'(v)\}$ and $j\in \{1,2,3,4\}$, we know $A \cap \left( \bigcup_{v\in V} M_v \right) =\emptyset$ for each defensive alliance $A\subseteq V'$ in~$G'$, i.e., $A\subseteq V$. Observe $(*)$: In $G'$, each $v\in V$ is incident with exactly $d'(v)$ negative edges.

For clarity, in the following we attach a possibly second subscript $G$ or $G'$ to $\deg$ and its variations to express of which graph we are talking about.
Let $A\subseteq V$ be a defensive alliance in the unsigned graph~$G$. By definition,  for all $v\in A$, \[\deg_{G,A}(v) + 1 \geq \deg_{G,\overline{A}}(v) = \deg_G(v)-\deg_{G,A}(v)\,,\] i.e.,  $\deg_{G,A}(v)\geq \frac{\deg_G(v)-1}{2}$. Since the degree is an integer, this is equivalent to \[\deg^+_{G',A}(v)=\deg_{G,A}(v)\geq \left\lceil\frac{\deg_G(v)-1}{2}\right\rceil = d'(v)-1\stackrel{(*)}{=} \deg_{G',\overline{A}}^-(v)-1\,.\] Also, $\deg_{G',A}^-(v)=0$. Therefore, $A$ is a defensive alliance in $G$ \iffl $A$ is a defensive alliance in~$G'$.
\end{proof}

As the set of vertices forming an alliance is the same in the original unsigned graph~$G$ and the constructed signed graph~$G'$, many complexity (hardness) results known for \textsc{Minimum Defensive Alliance} in unsigned graphs translate to complexity (hardness) results for  \textsc{Minimum Defensive Alliance} in signed graphs.
According to \cite{Enc2009,JamHedMcC2009}, \longversion{papers that are }discussing alliance problems on planar and chordal unsigned graphs, the following is a direct consequence from the reduction of \autoref{prop:From-unsigned-to-signed}. \longversion{Notice that t}\shortversion{T}he results for special graph classes \longversion{obtained this way }do not follow from \autoref{thm:def-alliable}.

\begin{theoremrep}\label{thm:MinDefAllNP}
\textsc{MinDefAll} is \NP-complete, even if the underlying unsigned graph is planar or chordal.
\end{theoremrep}

\begin{proof}
We still have to show that the presented reduction of \autoref{prop:From-unsigned-to-signed} preserves planarity and chordality. For planarity, it suffices to observe that only 4-cliques are attached to each vertex of a planar graph, and this can be done maintaining the planar embedding of the original planar graph.
Concerning chordality, we use the concept of a perfect elimination order. For a graph $G=(V,E)$; this refers to bijections $\sigma: \{1,\ldots,\vert V\vert\}\to V$ such that, for each $i\in \{1,\ldots,\vert V\vert\}$, $N[\sigma(i)]\cap\{\sigma(1), \ldots, \sigma(i)\}$ is a clique of $G[\{\sigma(1), \ldots, \sigma(i)\}]$. It is well known that a graph is chordal \iffl the graph has a perfect elimination order. As \textsc{Minimum Defensive Alliance} on unsigned chordal graphs is \NP-complete by \cite{JamHedMcC2009}, we can assume that the unsigned graph $G=(V,E)$ that we start with in our reduction is chordal and has a perfect elimination order. We now consider the graph $G'=(V',E')$ obtained by our reduction; more precisely, we discuss its underlying unsigned graph. We freely use the notations concerning the vertices of $G'$ from the previous description in the following. Since $N[v_{i,j}] = \{ v_{i,1}, v_{i,2}, v_{i,3}, v_{i,4}\}$ is a clique for each $v\in V$, $i\in \{ 1,\ldots,d'(v)\}$ and $j\in \{2,3,4\}$, these vertices can be mentioned in the end of a perfect elimination order. If we delete these vertices, $v_{i,1}$ is pendant and these vertices can follow in the order. The remaining vertices and edges are the vertices and edges from $G$. As $G$ is assumed to be chordal, we can conclude our definition of the ordering of the vertices of $G'$ by any perfect elimination ordering of the vertices of~$G$. By the reasoning above, this gives indeed a perfect elimination order of $G'$, proving that $G'$ is chordal.
\end{proof}

We now look into this problem from the viewpoint of parameterized complexity. While \textsc{Minimum Defensive Alliance} on unsigned graphs with solution-size parameterization is in \FPT,\longversion{ see \cite{FerRai07},}  %This is in contrast to what 
we will show for signed graphs a \W{1}-completeness result, where hardness is given by a reduction from \textsc{Clique} and membership is given by a reduction to \textsc{Short \longversion{Nondeterministic Turing Machine}\shortversion{NTM} Computation} \cite{Ces2003}. 
%The hardness reduction of the following theorem starts employs the problem \textsc{Clique} (with solution-size parameterization).

\begin{toappendix}
We now define the two problems that we employ more formally:
\problemdef{Clique}{}{An unsigned graph $G$ and an integer~$k$}{Is there a clique in $G$ with $k$ vertices?}

\problemdef{Short Nondeterministic Turing Machine Computation}{}{A single-tape nondeterministic Turing machine~$M$, a word~$x$ on the input alphabet~$\Sigma$ and an integer~$k$}{Does $M$ accept $x$ within $k$ steps?}
    
\end{toappendix}

\begin{theoremrep}\label{thm:MinDefAll-solutionsize}
\textsc{MinDefAll} with solution-size parameterization is \W{1}-complete, even on balanced signed graphs.
\end{theoremrep}

\begin{proof}
    It is well-known that \textsc{Clique} with solution-size parameterization is \W{1}-complete. This is the basis for our reduction. Let $\left \langle G=(V,E),k \right \rangle$ be any instance of \textsc{Clique}. Define $C_{v,i}=\{v_{ij}|j\in\{1,2,3,4\}\}$ for all $i\in\{1,\ldots,\max\{3,k\}\}$, and $V_E=\{\widehat{e}\, \mid e\in E\}$, also called edge-vertices. The reduction~$R$ is defined as: 
$R(\left \langle G,k \right \rangle):=$
\begin{itemize}
    \item [1)] Build a signed graph $G'=(V', E^{\prime\,+}, E^{\prime\,-})$:\\
    $V'=V \cup V_E \cup \left( \bigcup_{v \in V} \bigcup_{i=1}^k C_{v,i} \right) \cup \left( \bigcup_{\widehat{e} \in V_E} \bigcup_{i=1}^3 C_{\widehat{e},i} \right) $,\\
    $E^{\prime\,+}=\{u\widehat{e},v\widehat{e}\mid  e=uv \in E\}$,\\
    $E^{\prime\,-}=\{vv_{i1}, \widehat{e}\widehat{e}_{l1}, 
    v_{ij}v_{ih}, \widehat{e}_{lj}\widehat{e}_{lh}
    \mid v \in V,  i \in \{1,\ldots, k\},  e \in E,  
    l \in\{1,2,3\}, j, h\in \{1,2,3,4\}, j\neq h \}$.
    \item [2)] Return $\left \langle G', k+\frac{k(k-1)}{2} \right \rangle$.
\end{itemize}

If $\left \langle G,k \right \rangle$ is a \yes-instance of $\textsc{Clique}$, without loss of generality, $G$ contains a clique of size at least~$k$, denoted as $C_k=\{v_1,\dots, v_k\}$. Let $S=C_k\cup E_k$, where $E_k= \{\widehat{v_iv_j}\mid i, j \in \{1,\ldots, k\}, i\neq j\}$, combined with the reduction $R$, we know for each $v \in C_k$, $ \deg^+_S(v)=k-1$, so (1) $\deg^+_S(v)+1=k \geq \deg^-_S(v)=0$; (2) $\deg^+_S(v)+1=k\geq\deg^-_{\overline{S}}(v)=k$. For all $\widehat{e} \in E_k $, $\deg^+_S(\widehat{e})=2$, so (1) $\deg^+_S(\widehat{e})+1=3\geq \deg^-_S(\widehat{e})=0$; (2) $\deg^+_S(\widehat{e})+1=3\geq \deg^-_{\overline{S}}=3$. That is to say,  $S$ forms a defensive alliance of size $k+\frac{k(k-1)}{2}$.

Conversely, if the signed graph $G'$ built by~$R$ contains a defensive alliance $S$ with $|S|\leq k+\frac{k(k-1)}{2}$. For all $v \in C_{x,i}$, $\deg^+(v)+1=0+1<\lceil \frac{\deg^-(v)}{2} \rceil=\lceil \frac{3}{2} \rceil=2$. Therefore, such vertices cannot be in any defensive alliance.  Moreover, any two vertices in $V$ are not adjacent, and also $V_E$ forms an independent set. Therefore, $S=V_1 \cup V_2$, where $V_1\subseteq V, V_2 \subseteq  V_E$. By construction, we have $\deg^-(v)=k, $ for all $v \in V$ and $\deg^-(\widehat{e})=3$ for all $e \in E$. So for $v_1 \in V_1$, \[|\{\widehat{e} \in V_2\mid \exists x\in V: e=v_1x\in E\}|\geq k-1\,,\] and for $\widehat{e}\in V_2$, its two incident vertices must be in~$S$. So, we know $|V_1|\geq k$. To defend $k$ vertices in $V_1$, we need at least $(k-1)+(k-2)+\ldots+1=\frac{k(k-1)}{2}$ many corresponding edge-vertices in $V_2$ without the addition of new vertices. So $|S|=|V_1|+|V_2|\geq k+ \frac{k(k-1)}{2}$. Then we have $S=V_1\cup V_2$, where $|V_1|=k,$ $ |V_2|=\frac{k(k-1)}{2}$. Obviously, $V_1$ is a clique of $G$ with size $k$.

For the membership proof, we use the Turing machine characterization of \W{1} as developed in \cite{Ces2003}. Formally, we have to describe a (in our case, polynomial-time, as well as parameterized) reduction from \textsc{Minimum Defensive Alliance} to \textsc{Short Nondeterministic Turing Machine (NTM) Computation}, consisting in a single-tape nondeterministic Turing machine $M$, a word $x$ on the input alphabet of
$M$, and a positive integer $k$. The task is to decide if $M$ accepts~$x$ in at most~$k$ steps. Given a signed graph $G'=(V', E^{\prime\,+}, E^{\prime\,-})$ with $n$ vertices and $m=m^++m^-$ edges, and a positive integer $k'$, we can construct a single-tape nondeterministic Turing machine $M$, which nondeterministically guesses $k'$ vertices of $G'$ and writes as $(S\coloneqq)\{v_1,v_2,\ldots,v_{k'}\}$ onto the tape. Then $M$ scans each guessed vertex $v_i$ and rejects either (1) $\deg^+_S(v_i)+1<\deg^-_S(v_i)$, or (2) $\deg^+_S(v_i)+1<\deg^-(v_i)-\deg^-_S(v_i)$. Otherwise, it accepts. It can be verified that the Turing machine $M$ accepts in $\mathcal{O}(k'^2)$ steps \iffl there is a defensive alliance of size~$k'$ in~$G'$. This explains the correctness of the reduction that we sketched.
\end{proof}

\longversion{Notice that t}\shortversion{T}his result is not only a negative message, as by the well-known inclusion $\W{1}\subseteq\XP$, we also get the following algorithmic result, as $\XP$ is also known as `poor man's \FPT'.
\begin{corollary} $\textsc{MinDefAll}\in \XP$ when parameterized
with solution size.
\end{corollary}

According to [\cite{BliWol2018}, Theorem~$1$], \textsc{Minimum Defensive Alliance} on unsigned graphs with treewidth as its parameter is also  \W{1}-hard. With \autoref{prop:From-unsigned-to-signed}, this translates into our setting as follows.

\begin{theoremrep}
    \textsc{MinDefAll} is \W{1}-hard, when parameterized by the treewidth of the underlying unsigned graph.
\end{theoremrep}

\begin{proof}
    Let $(G,k)$ is an instance of \textsc{Minimum Defensive Alliance} on unsigned graphs, based on \autoref{prop:From-unsigned-to-signed}, we can obtain an instance of \textsc{Minimum Defensive Alliance} on signed graphs $(G',k)$ in polynomial time. We show that the treewidth of the underlying graph of $G'$ depends only on the treewidth of $G$, by modifying an optimal tree decomposition $\mathcal{T}$ of $G$ as follows: for each $v\in V(G)$, take an arbitrary node whose bag $B$ contains $v$, and then add $d'(v)$ nodes $C_1, C_2,\dots,C_{d'(v)}$ as its children such that the bag of each~$C_i$ is $B\cup \{v_{i,1}\}$; and then add a child node $N_i$ to each node~$C_i$ with bag $\{v_{i,1}, v_{i,2}, v_{i,3}, v_{i,4}\}$. It is easy to verify that the result is a valid tree decomposition of the underlying graph of $G'$ and its width is at most the treewidth of $G$ plus $1$ or $4$.

Following [\cite{BliWol2018}, Theorem $1$], we know \textsc{Minimum Defensive Alliance} on signed graphs is \W{1}-hard when parameterized by treewidth.
\end{proof}

Notice that the preceding result has some (negative) algorithmic consequences for a case that is important in the context of geographic applications, namely, that of planar signed networks. Not only is \textsc{MinDefAll} \NP-hard in this case by \autoref{thm:MinDefAllNP}, but also the standard approach to show that a planar graph problem is in \FPT, when parameterized by solution size, fails here, as it is based on a dynamic programming algorithm along a tree decomposition of small width, as show-cased for \textsc{Dominating Set} in \cite{Albetal02}.

Another typical restriction often discussed in the literature is an upper bound on the maximum degree of the graph. 
We have previously seen that the problems that we study are polynomial-time solvable on subcubic signed graphs; see \autoref{thm:subcubic-DA}. 
Currently, we do not know if this is also true if we lift the \longversion{upper bound on the}\shortversion{maximum} degree from three to four, but an increase to five changes the picture, as we show next\longversion{ by modifying the construction of \autoref{thm:def-alliable}}. %, based on a restriction of \textsc{NAE-3SAT} discussed in~\cite{DarDoc2020}.

\begin{theoremrep}\label{thm:MinDefAllOnMaxDegree5}
     \textsc{DefAll} (and \textsc{MinDefAll}) on signed graphs is \NP-complete even if the maximum degree is 5. 
\end{theoremrep}

\begin{proof}
%\todo[inline]{Wouldn't it suffice to describe the difference to the previous proof? Also, we start with \textsc{NAE-3SAT-E4}!}
    The membership follows by \autoref{thm:def-alliable}, also the hardness is obtained by reducing from \textsc{NAE-3SAT} very similarly. We only modify the clause gadget into \autoref{fig:MaxDeg5ClauseGadget}. 

    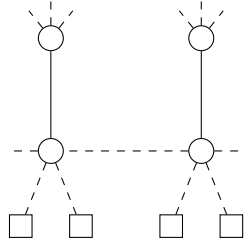
\begin{figure}[ht]
        	
    \centering
	    \begin{tikzpicture}[ transform shape]
			\tikzset{every node/.style={ fill = white,circle,minimum size=0.3cm}}
            
            \node[draw] (c1) at (0,1.5) {};
            \node[draw] (c2) at (2,1.5) {};
            \node[draw] (d1) at (0,0) {};
            \node[draw] (d2) at (2,0) {};
            \node[] (d0) at (-0.75,0) {};
            \node[] (d3) at (2.75,0) {};
            \node[] (x11) at (-0.4,2) {};
            \node[] (x12) at (0,2.25) {};
            \node[] (x13) at (0.4,2) {};
            \node[] (x21) at (2-0.4,2) {};
            \node[] (x22) at (2,2.25) {};
            \node[] (x23) at (2.4,2) {};
            
            \node[draw,rectangle] (d1n1) at (0.4,-1) {};
            \node[draw,rectangle] (d1n2) at (-0.4,-1) {};
            
            \node[draw,rectangle] (d2n1) at (1.6,-1) {};
            \node[draw,rectangle] (d2n2) at (2.4,-1) {};
            \path (c1) edge[-] (d1);
            \path (c2) edge[-] (d2);
            \path (d2) edge[dashed] (d1);
            \path (d1) edge[dashed] (d0);
            \path (d2) edge[dashed] (d3);
            
            \path (d1) edge[dashed] (d1n1);
            \path (d1) edge[dashed] (d1n2);
            \path (d2) edge[dashed] (d2n1);
            \path (d2) edge[dashed] (d2n2);
            \path (c1) edge[dashed] (x11);
            \path (c1) edge[dashed] (x12);
            \path (c1) edge[dashed] (x13);
            \path (c2) edge[dashed] (x21);
            \path (c2) edge[dashed] (x22);
            \path (c2) edge[dashed] (x23);
            
        \end{tikzpicture}

    \caption{Clause gadget: the squares represent the not-in-the-solution gadget.}
     \label{fig:MaxDeg5ClauseGadget}
\end{figure}

Let  $\phi=C_1\wedge\dots\wedge C_m$ be an input of \textsc{NAE-3SAT} with the variable set $X$, where $|X|=n$. In the same way as in the proof of \autoref{thm:def-alliable}, define $N_{v,i}=(V_{v,i},\emptyset,E^-_{v,i})$, where $V_{v,i}\coloneqq \{v_{i,1},v_{i,2},v_{i,3},v_{i,4}\}$, $i\in \mathbb{N^+}$, serving as not-in-the-solution gadgets. Notice that now, all these gadgets are small. Also, for each variable $x$ in~$X$: $C(x)=\{C_{j_1},\dots,C_{j_{n_x}}\mid x\in C_{j_i},i\in\{1,\dots,n_x\}, j_i\in\{1,\dots,m\}\}$,
    and $X'(x)\coloneqq\{x_h\mid h\in \{1,\dots,2n_x\}\} $. We construct the signed graph $G=(V,E^+,E^-)$ with:
    \begin{equation*}
        \begin{split}
            V={}& \{c_i,d_i\mid i\in \{1,\ldots,m\}\,\}
            \cup \{V_{d_i,l}\mid i\in \{1,\dots,m\}, l\in\{1,2\}\}\cup\bigcup_{x\in X} \left(X'(x)\cup \bigcup_{v\in X'(x), l\in \{1,2\}}V_{v,l}\right)\\
            E^+={}&\{c_jd_j\mid j\in\{1,\dots ,m\}\}\cup\{x_{2p-1}x_{2p} \mid p\in \{1,\dots,n_x\}, x\in X\}\\
            E^-={}&\{ d_id_{i+1},d_1d_m\mid i\in \{1,\ldots,m-1\}\}\cup\{d_iv_{d_i,l}, E^-_{d_i,l}\mid i\in\{1,\dots,m\},l\in\{1,2\}\}\cup{} \\&\left(\bigcup_{x\in X}\{c_{j_p}x_{2p-1}\mid p\in \{1,\dots,n_x\},C_{j_p}\in C(x)\}\right) \cup\{x_{2p}x_{2p+1}, x_{2n_x}x_1\mid x\in X, p\in \{1,\dots,n_x-1\} \} \cup{} \\&\left(\bigcup_{x\in X}\{x_{2h}v_{x_{2h},l}\mid h\in \{1,\dots,n_x\}, l\in \{1,2\}\}\right)\cup\left(\bigcup_{x\in X}\,\,\bigcup_{v\in X'(x), l\in \{1,2\}}E^-_{v,l}\right)
        \end{split}
    \end{equation*}

If~$A$ is a solution of $\phi$, then for each $C_i=x_{i1}\vee x_{i2}\vee x_{i3}$, there are $j_1,j_2 \in \{1,2,3\}$ such that $x_{ij_1}=1$ and $x_{ij_2}=0$. Let $S=\{x\in X\mid A(x)=0\}$. Define $D\coloneqq\{c_i,d_i\mid i\in \{1,\dots,m\}\}\cup \bigcup_{x\in S} X'(x)$, We show that $D$ is a defensive alliance. For each $c_i$,  1) $\deg^+_D(c_i)+1=2\geq\deg^-_D({c_i})$; 2) $\deg^+_D({c_i})+1=2\geq\deg^-_{\overline{D}}({c_i})$. For each $d_i$, $\deg^+_D(d_i)+1=2=\deg^-_D({c_i})=\deg^-_{\overline{D}}({c_i})$. For each vertex $x_h\in D$ associated to $x\in X$, 1) $\deg^+_D(x_h)+1=2\geq\deg^-_D(x_h)$; 2) $\deg^+_D(x_h)+1=2=\deg^-_{\overline{D}}(x_h)$. Therefore, $D$ is a defensive alliance.

Now assume there exists a defensive alliance $D\subseteq V$ of $G$. As for all $u\in V_{v,i}$, $\deg^-(u)\geq4>1=\deg^+(u) + 1$, we conclude that $D\subseteq \{c_i,d_i\mid i\in\{1,\dots,m\}\}\cup\{X'(x)\mid x\in X\}$. Now, assume that there exists some $d_i\in D$. Since $\deg^+(d_i)=1$ and $\deg^-(d_i)=4$, $\vert D \cap N^-(d_i)\vert =2$ and $ \vert D \cap N^+(d_i)\vert =1$. Because $v_{d_i,1},v_{d_i,2}\notin D$ $c_i,d_{i-1},d_{i+1} \in D$ (or $d_1,d_m\in D$ if $i\in \{1,m\}$). 
%By an inductive argument, 
Hence, $d_i\in D$ \iffl 
$C\coloneqq\{c_i,d_i \mid i\in \{1,\ldots,m\} \,\}\subseteq D$. If $c_i\in D$ for some~$i$, then $\deg^+(d_i)=1$ and $\deg^-(d_i)=3$ implies $\deg^+_{D}(c_i)=1$ and $\deg^-_{D}(c_i),\deg^-_{\overline{D}}(c_i)\in \{1,2\}$. With the argument from above, $C$ would be a subset of $D$. Assume there exists an $x\in X$ with $X'(x)\cap D\neq \emptyset$. By the proof of \autoref{thm:def-alliable}, $X'(x)\subseteq D$.

Define the assignment $A:X \to \{0,1\}$ with 
\[ x\mapsto \begin{cases}
    1, & x \in D\\
    0, & x \notin D
\end{cases}\]
As described before, if there exists a non-empty defensive alliance, then for each $i\in\{1,\ldots,m\}$, $c_i\in D$ and one or two negative neighbors. Therefore, $\phi$ is a \yes-instance of \textsc{NAE-3SAT}.
\end{proof}
\longversion{\noindent
}In conclusion, \textsc{(Min)DefAll} is \paraNP-hard when parameterized by maximum degree.
\begin{toappendix}
\begin{lemma}
The constructed signed graph $G$ has at most $16m$ many vertices.
\end{lemma}
\begin{proof}
    From the previous reduction construction, $| \{c_i,d_i\mid i\in \{1,\ldots,m\}\,\}
            \cup \{V_{d_i,l}\mid i\in \{1,\dots,m\}, l\in\{1,2\}\}|\leq 2m+8m=10m$. Combined with \autoref{lem:all-con}, there are  at most $10m+6m=16m$ many vertices in  total.
\end{proof}
\end{toappendix}
\longversion{\noindent}Naturally inherited from our reduction, we have the next ETH-based result:

\begin{corollary}
    Under ETH, there is no algorithm for solving \textsc{DefAll} (nor \textsc{MinDefAll}) on n-vertex signed graphs of maximum degree~$5$ in time $\mathcal{O}(2^{o(n)})$.
\end{corollary}

When we consider the combination of  different parameters, we can design parameterized algorithms with combined parameter solution size and maximum degree, as we have degree and diameter bounded by \autoref{prop:min-def-all-connected} and can hence obtain a (huge) kernel; see \cite{Fer1718}. Alternatively, we can build a search tree\longversion{ quite analogously to the case of unsigned graphs,} as in \cite{FerRai07}.

\begin{theorem}
\textsc{MinDefAll} is in \FPT when parameterized both by solution size and by maximum degree.
\end{theorem}

\begin{proof}
Let $(G,k)$ be an instance of \textsc{MinDefAll}, where the maximum degree of the underlying unsigned graph of $G$ is~$d$. We show that \textsc{MinDefAll} is fixed-parameter tractable with combined parameters $k$ and~$d$. From \autoref{prop:min-def-all-connected}, we are looking for a defensive alliance which is connected in the underlying graph. We guess a vertex~$v$ \longversion{being }in a defensive alliance~$D$ with $|D|\leq k$. If $\{v\}$ is already a\longversion{ defensive}\shortversion{n} alliance, we are done; otherwise, we branch on the subset of the neighborhood of~$v$, which leads to $2^d$ branches. Then we have a defensive alliance of size not larger than~$k$ or we keep on branching on the union of the neighborhoods resulting in a branch of size $2^{k(d-1)}$. As $D$ is connected, we can branch at most $k-1$ times and obtain a search-tree of size bounded by $2^d2^{k(k-2)(d-1)}=2^{\mathcal{O}(dk^2)}$. \longversion{Moreover, w}\shortversion{W}e can check \longversion{whether}\shortversion{if} we have already obtained a defensive alliance of size at most~$k$ in time $\mathcal{O}(k^2)$. We could branch over all vertices of~$G$ in the beginning, so that the problem can be solved in time $\mathcal{O}(k^22^{\mathcal{O}(dk^2)}n)$.
\end{proof}

We can also combine the parameters treewidth and maximum degree. Again, the picture changes \longversion{and we encounter parameterized tractability}\shortversion{towards \FPT}. 

\begin{theoremrep}
\textsc{(Min)DefAll} is in \FPT when parameterized both by treewidth and maximum degree.    
\end{theoremrep}

\begin{proof}
Let $G=(V, E^+, E^-)$ be an instance of \textsc{DefAll} with bounded treewidth $l$ and maximum degree $\Delta$. Actually, we will describe an algorithm based on the principle of dynamic programming that also solves the minimization version in the sense that it first decides if $G$ admits a defensive alliance at all (which is signalled by having computed $\infty$ as the size of the smallest alliance) or it will output the size of the smallest defensive alliance of~$G$. Clearly, this would also allow us to solve an instance of \textsc{MinDefAll} if necessary.  
Assume that  a nice tree decomposition $\mathcal{T}=(T,\{X_t\}_{t\in V(T)})$ of $G$ with width $l$ is known. Namely, recall that we can compute such a tree decomposition, given $G$, in \FPT-time when parameterized by the treewidth of the input (\cite{CygFKLMPPS2015}).
We consider $T$ as a rooted tree by arbitrarily choosing a root~$\rho$ and $T_t$ as the subtree of $T$ rooted at $t\in T$ defined by those vertices whose unique path to~$\rho$ will contain~$t$. Define $Y_t=\bigcup_{t'\in V(T_t)}X_{t'}$ for each $t\in V(T)$. We want to use dynamic programming to solve this problem. Therefore, we build a table~$\tau_t$ considering all partial solutions on $Y_t$ (for $t\in V(T)$) depending on the interaction between the partial solution and the bag~$X_t$. For each $A_t = \{v_1,\ldots,v_j\} \subseteq X_t$, let $\tau_t[A_t, \Delta^i_1, \Delta^e_1,\dots,\Delta^i_{j},\Delta^e_{j}]$ with $\vert A_t \vert \coloneqq j \leq l+1$, be the minimum size of the current solution $S_t$ on $Y_t$, where $X_t\cap S_t=A_t$ and $\Delta^i_{k}=\deg^+_{S_t}(v_k)-\deg^-_{S_t}(v_k)+1$ or $\Delta^e_{k}=\deg^+_{S_t}(v_k)-\deg^-_{\overline{S_t}}(v_k)+1=\deg_{S_t}(v_k)-\deg^-(v_k)+1$ records the first (\textbf{i}nternal) or second (\textbf{e}xternal) defensive alliance condition for each $k\in \{1, \ldots, j\}$, respectively. 
    
For each configuration $[A_t, \Delta^i_1, \Delta^e_1,\dots,\Delta^i_{j},\Delta^e_{j}]$ of node $t$, more precisely, its value $\tau_t[A_t, \Delta^i_1, \Delta^e_1,\dots,\Delta^i_{j},\Delta^e_{j}]$ is given by:
    \begin{equation*} 
        \begin{split}
        \min  \{|S_t|\mid{} & S_t\subseteq Y_t, S_t\cap X_t=A_t, \\ & \forall v_k\in A_t, \Delta^i_k =\deg^+_{S_t}(v_k)-\deg^-_{S_t}(v_k)+1 , \Delta^e_k = \deg_{S_t}(v_k)-\deg^-(v_k)+1, \\ &\forall u\in S_t\setminus A_t, \deg^+_{S_t}(u)-\deg^-_{S_t}(u)+1\geq 0, \deg_{S_t}(u)-\deg^-(u)+1\geq 0\}.
        \end{split}
    \end{equation*}
Here, we see already how the idea of a partial solution is implemented by using the fact that bags are separators in the graph. Namely, the vertices in $S_t$ that are not in $A_t$ must satisfy the conditions imposed upon defensive alliance vertices as they will never see any further neighbors when walking towards the root of the tree decomposition.
Notice that these conditions are not only locally looking at the part of the graph that has been processed when walking bottom-up along the tree decomposition, but the degree conditions are referring to the whole graph~$G$. This feature is different from most algorithms that use dynamic programming for computing graph parameters along tree decompositions.
Conversely, from the perspective of the vertex set $A_t$ presumed to be a subset of a defensive alliance, the exact set of neighbors in $Y_t\setminus X_t$ is irrelevant, it only counts how many friends or foes are within or outside the set $S_t\supseteq A_t$. This information is fixed in the last $2\cdot j$ descriptors of a table row.

Let $Z_t[A_t, \Delta^i_1, \Delta^e_1,\dots,\Delta^i_{j},\Delta^e_{j}]$ be the underlying set of sets, so that \[\tau_t[A_t, \Delta^i_1, \Delta^e_1,\dots,\Delta^i_{j},\Delta^e_{j}]=\min  \{|S_t|\mid S_t\in Z_t[A_t, \Delta^i_1, \Delta^e_1,\dots,\Delta^i_{j},\Delta^e_{j}]\}\,.\]
If $Z_t[A_t, \Delta^i_1, \Delta^e_1,\dots,\Delta^i_{j},\Delta^e_{j}]=\emptyset$, we set $\tau_t[A_t, \Delta^i_1, \Delta^e_1,\dots,\Delta^i_{j},\Delta^e_{j}]\coloneqq\infty$.  
    
    %$\textsc{State}$ stores the current state of solution $S_t$, i.e., $\textsc{State}=\texttt{True}$ if $S_t$ is already a defensive alliance, i.e., for all $k \in \{1,\ldots,j\}, ,\Delta^i_{k}\geq 0, \Delta^e_{k}\geq 0$; $\textsc{State}=\texttt{False}$ if there exists some vertex $v_k$ in $S_t$ that will never defend itself, i.e., $\deg_{S_t}(v_k)=\frac{\Delta^e_{k}-\Delta^i_{k}+\deg^-(v_k)}{2}>\deg^+(v_k)+1$; otherwise, $\textsc{State}=\texttt{Go}$ means it will keep on searching to form a defensive alliance. As most dynamic programming algorithms for the bounded treewidth graphs, we are going to construct the table of each $t$ by considering the tables of its children. Observe that if there are some entries with $\textsc{State}=\texttt{False}$, then it will not be considered in the parent table anymore.

Obviously, the number of entries of each table is bounded by $2^{l+1}\cdot(2\Delta+1)^{2(l+1)}$. We now specifically present how to build the table of $t\in V(T)$ by using the tables of its children. The initialization sets all tables of leaf nodes~$t$ of the tree decomposition with one entry $\tau_t[\emptyset]=0$.
    \begin{itemize}
        \item [-] Consider the case of an introduce node $t$ with unique child $t'$, where $X_t=X_{t'}\cup \{v\}$. For each configuration $[A_t, \Delta^i_1, \Delta^e_1,\dots,\Delta^i_{j},\Delta^e_{j}]$ of $t$, we have:
        \begin{itemize}
            \item [i)] If $v\notin A_t$, then $\tau_t[A_t,\Delta^i_1, \Delta^e_1,\dots,\Delta^i_{j},\Delta^e_{j}]=\tau_{t'}[A_t, \Delta^i_1, \Delta^e_1,\dots,\Delta^i_{j},\Delta^e_{j}]$.
            \item [ii)] If $v\in A_t$, without loss of generality, we regard $v$ as the $j$th vertex of $A_t$, for each $k\in \{1,\dots,j-1\}$, we define 
            \[ \Delta'^{,i}_{k}\coloneqq \begin{cases}
    \Delta^i_k-1, & \text{if }v_kv_{j}\in E^+ \\
    \Delta^i_k+1, & \text{if }v_kv_{j}\in E^-\\
    \Delta^i_k, & \text{otherwise}\\
\end{cases}\] and 
            \[ \Delta'^{,e}_{k}\coloneqq \begin{cases}
    \Delta^e_k-1, & \text{if }v_kv_{j}\in E^+\cup E^- \\
    \Delta^e_k, & \text{otherwise}\\
\end{cases}\] as well as 
\[\widetilde{\Delta}^i_k\coloneqq \deg^+_{A_t}(v_k)-\deg^-_{A_t}(v_k)+1\]
\[\widetilde{\Delta}^e_k\coloneqq \deg_{A_t}(v_k)-\deg^-(v_k)+1\]
for $k\in\{1,\dots,j\}$. Then, \[\tau_t[A_t,\Delta^i_1, \Delta^e_1,\dots,\Delta^i_{j},\Delta^e_{j}]=
            \begin{cases}
    \tau_{t'}[A_t\setminus\{v\}, \Delta'^{,i}_1, \Delta'^{,e}_1,\dots,\Delta'^{,i}_{j-1},\Delta'^{,e}_{j-1}]+1, &  \text{if}\quad 
    \begin{split}
        \Delta^i_{j}=& \widetilde{\Delta}^i_j \text{ and}\\ \Delta^e_{j}=&\widetilde{\Delta}^e_j
    \end{split}
    \\
     \infty, & \text{otherwise}
 \end{cases}\]
            
        \end{itemize}

        %For each entry ($\textsc{State}\neq\texttt{False}$) $t'[A_t, \Delta^i_1, \Delta^e_1,\dots,\Delta^i_{j},\Delta^e_{j}, \textsc{State}]$ of Table $t'$, we generate two entries for Table $t$ according to whether $v$ is in the solution: one is that $v$ is not in $S_t$, then $t[A_t,\Delta^i_1, \Delta^e_1,\dots,\Delta^i_{j},\Delta^e_{j}, \textsc{State}]=t'[A_t, \Delta^i_1, \Delta^e_1,\dots,\Delta^i_{j},\Delta^e_{j}, \textsc{State}]$; the other is that $v$ is in $S_t$, we initialize $\Delta^i_{j+1}=1,\Delta^e_{j+1}=1-\deg^-(v)$ to denote the  two defensive conditions of $v$, then update each $\Delta^i_{j}=\Delta^i_{j}+1,\Delta^e_{j}=\Delta^e_{j}+1$, $\Delta^i_{j+1}=\Delta^i_{j+1}+1,\Delta^e_{j+1}=\Delta^e_{j+1}+1$ if $v_j$ is positively connected to $v$; $\Delta^i_{j}=\Delta^i_{j}-1,\Delta^e_{j}=\Delta^e_{j}+1$, $\Delta^i_{j+1}=\Delta^i_{j+1}-1,\Delta^e_{j+1}=\Delta^e_{j+1}+1$ if $v_j$ is negatively connected to $v$. In succession, update the \textsc{State}, if the new $\textsc{State}\neq\texttt{False}$, then $t[A_t\cup\{v\}, \Delta^i_1, \Delta^e_1,\dots,\Delta^i_{j + 1},\Delta^e_{j + 1}, \textsc{State}]=t'[A_t, \Delta^i_1, \Delta^e_1,\dots,\Delta^i_{j},\Delta^e_{j}, \textsc{State}]+1$. 

        \item [-] Now, consider the case of a forget node $t$ with unique child $t'$, where $X_t=X_{t'}\setminus\{v\}$.  
        For each configuration $[A_t, \Delta^i_1, \Delta^e_1,\dots,\Delta^i_{j},\Delta^e_{j}]$ of $t$, we can compute the table entry $\tau_t[A_t, \Delta^i_1, \Delta^e_1,\dots,\Delta^i_{j},\Delta^e_{j}]$ as
        \[\min \{\tau_{t'}[A_t, \Delta^i_1, \Delta^e_1,\dots,\Delta^i_{j},\Delta^e_{j}],\tau_{t'}[A_t\cup\{v\},\Delta^i_1, \Delta^e_1,\dots,\Delta^i_{j},\Delta^e_{j},k^i,k^e]\mid k^i,k^e\geq 0\}
        \]
        
        %For each entry ($\textsc{State}\neq\texttt{False}$) $t'[A_t, \Delta^i_1, \Delta^e_1,\dots,\Delta^i_{j},\Delta^e_{j}, \textsc{State}]$ of Table $t'$, we generate the new entry of Table $t$: if $v\notin A_t$, then $t[A_t,\Delta^i_1, \Delta^e_1,\dots,\Delta^i_{j},\Delta^e_{j}, \textsc{State}]=t'[A_t, \Delta^i_1, \Delta^e_1,\dots,\Delta^i_{j},\Delta^e_{j}, \textsc{State}]$; if $v\in A_t$, without loss of generality, $v$ can be regarded as the $j$th vertex of $A_t$, only when $\Delta^i_{j}\geq ,\Delta^e_{j}\geq 0$, $t[A_t\setminus\{v\},\Delta^i_1, \Delta^e_1,\dots,\Delta^i_{j-1},\Delta^e_{j-1}, \textsc{State}]=t'[A_t, \Delta^i_1, \Delta^e_1,\dots,\Delta^i_{j},\Delta^e_{j}, \textsc{State}]$. Note that in this case, it is possible for two entries to be generated with the same items, we should of course keep the one with the smaller solution size. 

        \item [-] Finally, consider the case of a join node $t$ with two children $t'$ and $t''$, where $X_t=X_{t'}=X_{t''}$.
        For each configuration $[A_t, \Delta^i_1, \Delta^e_1,\dots,\Delta^i_{j},\Delta^e_{j}]$ of $t$, we set:
        \begin{equation*}
        \begin{split}
            \tau_t[A_t, \Delta^i_1, \Delta^e_1,\dots,\Delta^i_{j},\Delta^e_{j}] = \min & \{\tau_{t'}[A_t, x^i_1, x^e_1, \dots, x^i_{j},x^e_{j}] + \tau_{t''}[A_t, x^i_1, x^e_1, \dots, x^i_{j},x^e_{j}] - \vert A_t \vert \\ & \mid x_k^i+ y_k^i = \Delta^i_k +\widetilde{\Delta}_k^i,\, x_k^e + y_k^e = \Delta^e_k + \widetilde{\Delta}_k^e\}
        \end{split}
        \end{equation*}
        
        %We consider the entry ($\textsc{State}'\neq\texttt{False}$) $t'[A_t, \Delta'^{,i}_1, \Delta'^{,e}_1,\dots,\Delta'^{,i}_{j},\Delta'^{,e}_{j}, \textsc{State}']$ of Table $t'$ and the entry ($\textsc{State}''\neq\texttt{False}$) $t''[A_t, \Delta''^{,i}_1, \Delta''^{,e}_1,\dots,\Delta''^{,i}_{j},\Delta''^{,e}_{j}, \textsc{State}'']$ of Table $t''$ (with the same $A_t$). Without loss of generality, the sequence of vertices of $A_t$ in these two entries can be regarded as the same. So if the remaining items in both entries are the same, then $t[A_t,\Delta^i_1, \Delta^e_1,\dots,\Delta^i_{j},\Delta^e_{j}, \textsc{State}]=t'[A_t, \Delta'^{,i}_1, \Delta'^{,e}_1,\dots,\Delta'^{,i}_{j},\Delta'^{,e}_{j}, \textsc{State}'](=t''[A_t, \Delta''^{,i}_1, \Delta''^{,e}_1,\dots,\Delta''^{,i}_{j},\Delta''^{,e}_{j}, \textsc{State}''])$; otherwise, for each $v_j\in A_t$, update $\Delta^i_{j}=\Delta'^{,i}_{j}+\Delta''^{,i}_{j}-\Delta^i_{j}(A_t),\Delta^e_{j}=\Delta'^{,e}_{j}+\Delta''^{,e}_{j}-\Delta^e_{j}(A_t)$, where $\Delta^i_{j}(A_t)=\deg^+_{A_t}(v_j)-\deg^-_{A_t}(v_j)+1, \Delta^e_{j}(A_t)=\deg_{A_t}(v_j)-\deg^-(v_j)+1$. After checking the current \textsc{State}, we have $t[A_t,\Delta^i_1, \Delta^e_1,\dots,\Delta^i_{j},\Delta^e_{j}, \textsc{State}]=t'[A_t, \Delta'^{,i}_1, \Delta'^{,e}_1,\dots,\Delta'^{,i}_{j},\Delta'^{,e}_{j}, \textsc{State}']+t''[A_t, \Delta''^{,i}_1, \Delta''^{,e}_1,\dots,\Delta''^{,i}_{j},\\ \Delta''^{,e}_{j}, \textsc{State}'']-|A_t|$. 
    \end{itemize}
For the remaining proof, we want to show that the given recursion fits the definition. For this we use induction on the tree decomposition structure. We start with a leaf node $t$. This implies $X_t = Y_t = \emptyset$. Hence, $A_t=\emptyset$ and $S_t=\emptyset$ fulfills all necessary conditions of the underlying set of $\tau_t[\emptyset]$. This implies $\tau_t[\emptyset]=0$.

Let now $t$ be a node and assume we calculated all values for the tables of the children. We will use case distinction for $t$ depending on three cases, $t$ being either an introduce, a forget or a join node. 

We start with the introduce node case. Let $t'$ be the child of $t$ and $v$ be the introduced vertex. $A_t=\{v_1,\ldots,v_j\}\subseteq X_t$ and $\Delta_1^i, \Delta_1^e, \ldots, \Delta_j^i, \Delta_j^e\in \{-\Delta+1,\ldots, \Delta+1\}$. Assume $v\notin A_t$. Hence, $A_t \subseteq X_{t'}$ and as $Y_t = Y_{t'} \cup \{v\}$, $S_t \subseteq Y_{t'}$ for each $S_t \in Z_{t}[A_t,\Delta_1^i, \Delta_1^e, \ldots, \Delta_j^i, \Delta_j^e]$. Therefore, $Z_t[A_t,\Delta_1^i, \Delta_1^e, \ldots, \Delta_j^i, \Delta_j^e] = Z_{t'}[A_t,\Delta_1^i, \Delta_1^e, \ldots, \Delta_j^i, \Delta_j^e]$ and $\tau_t[A_t,\Delta_1^i, \Delta_1^e, \ldots, \Delta_j^i, \Delta_j^e] = \tau_{t'}[A_t,\Delta_1^i, \Delta_1^e, \ldots, \Delta_j^i, \Delta_j^e]$.

Let $v\in A_t$. Without loss of generality, $v=v_j$. Since $N(v)\cap Y_t \subseteq X_t$ by the separator property of bags, $\Delta^i_{j} \neq \deg^+_{A_t}(v_j)-\deg^-_{A_t}(v_j)+1=\widetilde{\Delta}^i_j$ or $\Delta^e_{j} \neq \deg_{A_t}(v_j)-\deg^-(v_j)+1=\widetilde{\Delta}^e_j$ would imply $Z_t[A_t,\Delta_1^i, \Delta_1^e, \ldots, \Delta_j^i, \Delta_j^e] = \emptyset$ and $\tau_t[A_t,\Delta_1^i, \Delta_1^e, \ldots, \Delta_j^i, \Delta_j^e] = \infty$. Therefore, assume $\Delta^i_{j}= \widetilde{\Delta}^i_j$ and $\Delta^e_{j}=\widetilde{\Delta}^e_j$.
     Let $S_t\in Z_t[A_t,\Delta_1^i, \Delta_1^e, \ldots, \Delta_j^i, \Delta_j^e]$ and $S_{t'} \coloneqq S_t \setminus \{v\}$ with $ \vert S_t \vert = \tau_t[A_t,\Delta_1^i, \Delta_1^e, \ldots, \Delta_j^i, \Delta_j^e]$. For all $k\in \{1, \ldots, j-1\}$, $\Delta'^{,i}_{k} = \deg^+_{S_{t'}}(v_k)-\deg^-_{S_{t'}}(v_k)+1$ and $\Delta'^{,e}_{k}=\deg_{S_{t'}}(v_k)-\deg^-(v_k)+1$ holds by definition of $\Delta'^{,i}_k, \Delta'^{,e}_k$. For each $u \in S_t\setminus X_t= S_{t'} \setminus X_{t'}$, $\deg^+_{S_t}(u)=\deg^+_{S_{t'}}(u)$ and $\deg^-_{S_t}(u)=\deg^-_{S_{t'}}(u)$. Since further $S_{t'}\subseteq Y_{t'}$,  $S_{t'}\in Z_{t'}[A_t\setminus \{v\},\Delta'^{,i}_1, \Delta'^{,e}_1, \ldots, \Delta'^{,i}_{j-1}, \Delta'^{,e}_{j-1}]$. Assume there exists an $S'_{t'}\in Z_{t'}[A_t\setminus \{v\},\Delta'^{,i}_1, \Delta'^{,e}_1, \ldots, \Delta'^{,i}_{j-1}, \Delta'^{,e}_{j-1}]$ with $\vert S'_{t'}\vert < \vert S_{t'}\vert$. This would contradict $ \vert S_t \vert = \tau_t[A_t,\Delta_1^i, \Delta_1^e, \ldots, \Delta_j^i, \Delta_j^e]$, $S'_{t'}\cup \{v\} \in Z_t[A_t,\Delta_1^i, \Delta_1^e, \ldots, \Delta_j^i, \Delta_j^e]$ and $\vert S'_{t'}\cup \{v\} \vert < \vert S_{t'}\cup \{v\} \vert $.

     As our second case, assume that $t$ is a forget node and that $t'$ is its child with $X_t=X_{t'}\setminus\{v\}$. We know that $Y_t=Y_{t'}$. Let $A_t=\{v_1,\ldots,v_j\}\subseteq X_t$ and $\Delta_1^i, \Delta_1^e, \ldots, \Delta_j^i, \Delta_j^e\in \{-\Delta+1,\ldots, \Delta+1\}$,  for each $S_t\subseteq Y_t (=Y_{t'})$ with $S_t\cap X_t=A_t$, if $v\notin S_{t}$, then $Z_{t'}[A_t,\Delta_1^i, \Delta_1^e, \ldots, \Delta_j^i, \Delta_j^e]\subseteq Z_t[A_t,\Delta_1^i, \Delta_1^e, \ldots, \Delta_j^i, \Delta_j^e]$; if $v\in S_t$, as $v\notin A_t$, i.e., $v\in S_t\setminus A_t$, only when $k^i \coloneqq \deg^+_{S_t}(v)-\deg^-_{S_t}(v)+1 \geq 0,\, k^e\coloneqq\deg_{S_t}(v)-\deg^-(v)+1 \geq 0$, we have $Z_{t'}[A_t\cup\{v\}, \Delta^i_1, \Delta^e_1,\dots,\Delta^i_{j},\Delta^e_{j},k^i,k^e]\subseteq Z_t[A_t,\Delta_1^i, \Delta_1^e, \ldots, \Delta_j^i, \Delta_j^e]$. So $\tau_t[A_t, \Delta^i_1, \Delta^e_1,\dots,\Delta^i_{j},\Delta^e_{j}]=\min \{t'[A_t, \Delta^i_1, \Delta^e_1,\dots,\Delta^i_{j},\Delta^e_{j}],t'[A_t\cup\{v\},\Delta^i_1, \Delta^e_1,\dots,\Delta^i_{j},\Delta^e_{j},k^i,k^e]\mid k^i,k^e\geq 0\}
        $.

        Finally, assume that $t$ is a join node with two children $t'$ and $t''$, where $X_t=X_{t'}=X_{t''}$. We know that $Y_t=Y_{t'}\cup Y_{t''}$, from the property of tree decomposition, $Y_{t'}\cap Y_{t''}=X_t=X_{t'}=X_{t''}$. Let $A_t=\{v_1,\ldots,v_j\}\subseteq X_t$ and $\Delta_1^i, \Delta_1^e, \ldots, \Delta_j^i, \Delta_j^e\in \{-\Delta+1,\ldots, \Delta+1\}$,  for each $S_t\subseteq Y_t$ with $S_t\cap X_t=A_t$, let $S_t=S_1\cup S_2\cup S_3$, where $S_1\subseteq S_t\subseteq Y_{t'}\setminus Y_{t''}, S_2\subseteq S_t\subseteq Y_{t''}\setminus Y_{t'}, S_3=S_t\setminus (S_1\cup S_2)$. Clearly, $S_{t'}\coloneqq S_1\cup S_3\subseteq Y_{t'}$ with $S_{t'}\cap X_{t'}=S_3$,  $S_{t''}\coloneqq S_2\cup S_3\subseteq Y_{t''}$ with $S_{t''}\cap X_{t''}=S_3$,  $S_3= S_t\cap (Y_{t'}\cap Y_{t''})=A_t$. So for each $v_k\in A_t$, \[\Delta^i_{k}=\deg^+_{S_t}(v_k)-\deg^-_{S_t}(v_k)+1 =\deg^+_{S_1}(v_k)+\deg^+_{S_2}(v_k)+\deg^+_{S_3}(v_k)-\deg^-_{S_1}(v_k)-\deg^-_{S_2}(v_k)-\deg^-_{S_3}(v_k)+1\,,\] i.e., $\Delta^i_{k}=(\deg^+_{S_1}(v_k)+\deg^+_{S_3}(v_k)-\deg^-_{S_1}(v_k)-\deg^-_{S_3}(v_k)+1)+(\deg^+_{S_2}(v_k)+\deg^+_{S_3}(v_k)-\deg^-_{S_2}(v_k)-\deg^-_{S_3}(v_k)+1)-(\deg^+_{S_3}(v_k)-\deg^-_{S_3}(v_k)+1)=\Delta'^{,i}_{k}+\Delta''^{,i}_{k}-\widetilde{\Delta^i_{k}}$, and analogously, $\Delta^e_{k}=\Delta'^{,e}_{k}+\Delta''^{,e}_{k}-\widetilde{\Delta}^e_{k}$. Therefore, $Z_{t'}[A_t, \Delta'^{,i}_1, \Delta'^{,e}_1, \dots,\Delta'^{,i}_{j},\Delta'^{,e}_{j}]\cup Z_{t''}[A_t, \Delta''^{,i}_1, \Delta''^{,e}_1, \dots, \Delta''^{,i}_{j},\Delta''^{,e}_{j}]\subseteq Z_t[A_t, \Delta^i_1, \Delta^e_1,\dots,\Delta^i_{j},\Delta^e_{j}]$, as \[|S_t|=|S_1|+|S_2|+|S_3|=(|S_1|+|S_3|)+(|S_2|+|S_3|)-|S_3|\,,\] $ \tau_t[A_t, \Delta^i_1, \Delta^e_1,\dots,\Delta^i_{j},\Delta^e_{j}]$ equals the minimum over all sums \[\tau_{t'}[A_t, \Delta'^{,i}_1, \Delta'^{,e}_1, \dots,\Delta'^{,i}_{j},\Delta'^{,e}_{j}] + \tau_{t''}[A_t, \Delta''^{,i}_1, \Delta''^{,e}_1, \dots, \Delta''^{,i}_{j},\Delta''^{,e}_{j}] - \vert A_t \vert\] where $\Delta'^{,i}_{k}+\Delta''^{,i}_{k}=\Delta^i_{k}+\widetilde{\Delta}^i_{k},\Delta'^{,e}_{k}+\Delta''^{,e}_{k}=\Delta^e_{k}+\widetilde{\Delta}^e_{k}$, as we claimed it above.
\end{proof}

\begin{toappendix}
Admittedly, networks of bounded degree are rarely encountered in the real world. Still, combinatorial questions restricted to bounded-degree graphs are traditionally considered, as this study can serve as an indication for conceptually related graph classes. For instance, although high-degree hubs could occur in real-world networks, they still have bounded average degree or bounded degeneracy, notions related to bounded degree. We have to leave it as an open question if we can generalize the previous result to the combined parameters treewidth plus average degree or plus degeneracy. We can only state that our algorithmic result also leads to an \FPT-result concerning the parameter \emph{domino treewidth}, see \cite{BodEng97,Bod99,DinOpo95}.
\end{toappendix}

\shortversion{Our algorithm uses dynamic programming on tree decomposition, but the table size depends 
%but the number of rows per table crucially depend 
on the maximum degree.}% of the graph.}
%\todo[inline]{Supposedly, we can afford an algorithm sketch in the main part, but the correctness proof is only in the long version.}

This result \longversion{hence }raises the question if we can find efficient algorithms for sparse graphs of bounded treewidth that might model at least some realistic network scenarios.

We next discuss a new structural parameter for signed graphs, \emph{signed neighborhood diversity}~\snd, which makes the problems tractable, parameterized by~\snd.  Our construction is a non-trivial extension of a similar result for unsigned graphs shown in \cite{GaiMaiTri2021}.

\begin{definition}
    \textbf{(Signed Neighborhood Diversity~\snd)} Let $G=(V,E^+,E^-)$ be a signed graph. Define the binary relation $\equivnd$ on~$V$ by $v\equiv u$ \iffl $N^+(v) \setminus \{u\} = N^+(u) \setminus \{v\}$ and $N^-(v) \setminus \{u\} = N^-(u) \setminus \{v\}$. $\snd(G)$ is  the number of equivalence classes of $\equivnd$ on~$G$.
\end{definition}
We will exhibit the parameterized algorithm in the following, which is implemented by an ILP with $2\cdot \snd(G)$ many variables. By \cite[Theorem 6.5]{CygFKLMPPS2015}, this ILP can be solved in \FPT-time.

Let $d \coloneqq \snd(G)$, and $C_1,\ldots,C_d$ denote the equivalence classes of $\equivnd$. Further, define $N^+_i\coloneqq \{ j\in \{ 1, \ldots, d\} \mid C_j \subseteq N^+(C_i)\}$ and $N^-_i\coloneqq \{ j\in \{ 1, \ldots, d\} \mid C_j \subseteq N^-(C_i)\}$. For the ILP, we need $z^+_1,\ldots,z^+_d\in \{0,1\}$ with $z^+_i=1$  \iffl $i\in N_i^+$. Similarly, define $z^-_1,\ldots,z^-_d\in \{0,1\}$.
The ILP has the variables, $x_1, \ldots, x_d, w_1, \ldots, w_d$. For each $i\in \{1, \ldots, d\}$, $x_i\in \{0,\ldots,\vert C_i \vert\}$ represents the number of vertices in $C_i$ which are in our solution; $w_i \in \{0, 1\}$ expresses if $x_i$ is 0 or not.
Then the ILP is given by:%\todo{do we really need the numbered equations in the short version? Without we could pull the forall i out above and safe 2 lines, maybe 3 if you merge the two last}
\longversion{
\begin{align}
     ( \sum_{j\in N^+_i}x_{j}) + 1 - z_i^+ &\geq ( \sum_{j\in N^-_i}x_{j}) - z_i^- -2 \cdot (1-w_i) \cdot \vert V \vert \nonumber\\ \forall i&\in \{ 1,\ldots, d\} \label{con:def_all_inside}\\
     ( \sum_{j\in N^+_i}x_{j}) + 1 - z_i^+ & \geq ( \sum_{j\in N^-_i}\vert C_j\vert - x_{j}) - 2 \cdot (1-w_i) \cdot \vert V \vert \nonumber\\\forall i &\in \{ 1,\ldots, d\}  \label{con:def_all_outside}\\
    x_{i} &\leq \vert C_i \vert \cdot w_{i} \quad\forall i\in \{ 1,\ldots, d\}\label{con:w1if}\\
    w_{i} &\leq  x_{i} \qquad\quad\forall i\in \{ 1,\ldots, d\}.\label{con:w0if}
\end{align}
}
\shortversion{
\begin{align*}
     ( \sum_{j\in N^+_i}x_{j}) + 1 - z_i^+ &\geq ( \sum_{j\in N^-_i}x_{j}) - z_i^- -2 \cdot (1-w_i) \cdot \vert V \vert \nonumber \\
     ( \sum_{j\in N^+_i}x_{j}) + 1 - z_i^+ & \geq ( \sum_{j\in N^-_i}\vert C_j\vert - x_{j}) - 2 \cdot (1-w_i) \cdot \vert V \vert \nonumber\\
    w_{i} \leq {}& x_{i} \leq \vert C_i \vert \cdot w_{i}, \qquad\qquad\forall i\in \{ 1,\ldots, d\}\label{con:w}
\end{align*}
}
\begin{lemmarep}
    There exists a defensive alliance $S \subseteq V$ on $G$ \iffl the ILP has a solution. If \longversion{such }a solution exists, \longversion{then }$\vert S\vert = \sum_{i=1}^d x_i$.  
\end{lemmarep}

\begin{proof}
    Let $S$ be a defensive alliance on $G$. Then define $x'_i \coloneqq \vert S\cap C_i \vert$, for each $i\in \{ 1,\ldots, d\}$. Let $i\in \{1,\ldots,d\}$. Then set $w'_i$ to 1 \iffl $x'_i \neq 0$. Otherwise, set $w_i$ to 0. Clearly, $w'_i \in \{0,1\}$ and $0\leq x'_i\leq \vert C_i\vert$. Further, $x'_i$ and $w'_i$ fulfill the inequalities in\shortversion{ the last line of the definition of the ILP.} \longversion{\ref{con:w1if} and \ref{con:w0if}.}  
       
    Now we want to show that $\deg^+_S(v) = \left( \sum_{j\in N^+_i} x'_{j} \right) - z_i^+$ holds for each $v \in S \cap C_i$. Since, for each $j\in \{1,\ldots,d\}$ and $v,u\in C_j$, $N^+(v) \setminus \{u\} = N^+(u) \setminus \{v\}$, $C_k \subseteq N^+(C_i)$ or $C_k \cap N^+(C_i) = \emptyset$. Therefore, $\deg^+_{S\setminus C_i}(v) = \sum_{j\in N^+_i\setminus \{ i \}}x'_{j}$ for each $v\in C_i$. Let $v\in S\cap C_i$. If $i\in N_i^+$, then $\deg^+_{S \cap C_i}(v) = \vert S\cap C_i \vert - 1 = x'_i - z_i$. This implies $\deg^+_S(v) = \left( \sum_{j\in N^+_i} x'_{j} \right) - z_i^+$ for all $v \in S \cap C_i$.        Analogously, $\deg^-_S(v)=\left( \sum_{j\in N^-_i}x'_{j}\right)  - z_i^-$ and $ \deg_{\overline{S}}(v)=\left( \sum_{j\in N^-_i}\vert C_j\vert - x'_{j}\right) $ (as $\deg^-_{C_j\setminus S}(v) = \vert C_j \setminus S \vert = \vert C_j \vert - x'_j$ for each $j\in \{1, \ldots, d\}$) holds for each $v\in S \cap C_i$. 

    If $w'_i=0$, then we observe
        \begin{equation*}
            \begin{split}
                ( \sum_{j\in N^+_i}x'_{j}) + 1 - z_i^+ & \geq 0 > \vert V\vert - 2 \vert V\vert\geq ( \sum_{j\in N^-_i}x'_{j}) - z_i^- -2 \cdot (1-w_i) \cdot \vert V \vert.
            \end{split}
        \end{equation*}
        Analogously,  $\left( \sum_{j\in N^+_i}x'_{j}\right) + 1 - z_i^+ > \left( \sum_{j\in N^-_i}\vert C_j\vert - x'_{j}\right) - 2 \cdot (1-w_i)\cdot  \vert V \vert$ holds. For $w'_i=1$, there exists a $v\in S\cap C_i$. Then 
        \begin{equation*}
            \begin{split}
                 ( \sum_{j\in N^+_i}x'_{j}) + 1 - z_i &=\deg^+_{S}(v) \geq \deg^-_{S}(v) 
                 \geq ( \sum_{j\in N^-_i}x'_{j}) - z_i^- 
                 \geq ( \sum_{j\in N^-_i}x'_{j}) - z_i^- -2 \cdot (1-w_i) \cdot \vert V \vert.
            \end{split}
        \end{equation*}
        Furthermore, 
        \begin{equation*}
            \begin{split}
                 ( \sum_{j\in N^+_i}x'_{j}) + 1 - z_i &=\deg^+_{S}(v) \geq \deg^-_{\overline{S}}(v) 
                 \geq ( \sum_{j\in N^-_i}\vert C_j \vert - x'_{j})  
                 \geq ( \sum_{j\in N^-_i}\vert C_j \vert - x'_{j}) - 2 \cdot (1-w_i)\cdot  \vert V \vert.
            \end{split}
        \end{equation*}
        Hence $x'_1,\ldots, x'_d,w'_1,\ldots,w'_d$ is a solution for the ILP. Further \[\vert S \vert = \sum_{j=1}^d \vert S \cap C_j\vert = \sum_{j=1}^d x'_j\].

        Now assume there exists a solution $x'_1,\ldots,x'_d,w'_1,\ldots,w'_d$ for the ILP. Let $i\in \{1,\ldots, d\}$ If $x'_i =0$, then $0 \leq w'_i \leq x'_i \leq 0$\longversion{ (see inequality \ref{con:w0if})} implies $w'_i=0$. For $x'_i \neq 0$, $ 0 < x'_i \leq \vert C_i\vert\cdot w'_i$ implies $w'_i = 1$.

        Choose a set $S$ such that $x_j=\vert S\cap C_j \vert$ for $j \in \{1, \ldots, d\}$. This can be done as $0\leq x_j\leq \vert C_j\vert$ and $C_1, \ldots, C_d$ is a partition. As above, $\deg^+_S(v) = \left( \sum_{j\in N^+_i} x'_{j} \right) - z_i^+$, $\deg^-_S(v)=\left( \sum_{j\in N^-_i}x'_{j}\right)  - z_i^-$ and $ \deg_{\overline{S}}(v)=\left( \sum_{j\in N^-_i}\vert C_j\vert - x'_{j}\right) $ hold for all $i\in \{1,\ldots, d\}$ and $v\in C_i \cap S$. Let $i \in \{1,\ldots,d\}$ and $v\in S \cap C_i$. This implies $x'_i>0$ and $w'_i=1$. Therefore,
        \begin{equation*}
            \begin{split}
                \deg^+_{S}(v) + 1 &= ( \sum_{j\in N^+_i} x'_{j} ) - z_i^+ + 1 \geq ( \sum_{j\in N^-_i}\vert C_j \vert - x'_{j} ) - 2 \cdot (1 - w_i) \cdot \vert V \vert
                = ( \sum_{j\in N^-_i}\vert C_j \vert - x'_{j} ) = \deg_{\overline{S}}^-(v) 
            \end{split}
        \end{equation*}
    and
        \begin{equation*}
            \begin{split}
                \deg^+_{S}(v) + 1 &= ( \sum_{j\in N^+_i} x'_{j} ) - z_i^+ + 1 \geq ( \sum_{j \in N^-_i} x'_{j} ) - z_i^- - 2 \cdot (1 - w_i) \cdot \vert V \vert
                = ( \sum_{j\in N^-_i} x'_{j}) - z_i^- = \deg_S^-(v). 
            \end{split}
        \end{equation*}
        Therefore, $S$ is a defensive alliance.
\end{proof}

\begin{theorem}\label{thm:snd-FPT}
    $\snd$-\textsc{Minimum Defensive Alliance} and $\snd$-\textsc{Defensive Alliability} are in \FPT.
\end{theorem}

As a side note, observe that the ILP formulation above also contains, as a sub-system of equations, a way how to express 
\textsc{(Min)DefAll} as an ILP, which could be interesting for solving such problems in practice.

We were not able to find a similar parameterized algorithm when considering the parameterization by the neighborhood diversity of the underlying unsigned graph. This is also due to the fact that we do not know if \textsc{Minimum Defensive Alliance} is \longversion{still hard or possibly }polynomial-time solvable if the underlying graph is a clique. \longversion{Knowing good algorithmic approaches for \textsc{Minimum Defensive Alliance} on cliques could be also helpful when thinking of modular decompositions or similar concepts that are often used for developing efficient algorithms on networks. }Therefore, it is natural to consider weaker parameterizations of the underlying unsigned graph. This discussion makes the following result interesting.

\begin{theorem}
\textsc{DefAll} and \textsc{MinDefAll} can be solved in \FPT-time when parameterized by the vertex cover number of the underlying unsigned graph.
\end{theorem}

\begin{proof}
    Let~\vc be the vertex cover number of the underlying graph of~$G$ and $C$ be a minimum vertex cover of the underlying graph. Notice that we can compute such a set~$C$ in \FPT-time; the currently best algorithm is described in \cite{CheKanXia2010}.  Then, $I\coloneqq V\setminus C$ is an independent set. For each vertex in~$I$, all its neighbors are in~$C$. \longversion{Distinguished from attributes of edge relationship}\shortversion{Clearly}, there are at most $3^{\vc}$ different neighborhoods for the vertices of~$I$. So \longversion{the signed neighborhood diversity }$\snd(G)\leq 3^{\vc}+\vc$. From \autoref{thm:snd-FPT}, the claim follows. %\textsc{Minimum Defensive Alliance} and \textsc{Defensive Alliability} parameterized by the vertex cover number of the underlying unsigned graph can also be solved in \FPT-time. 
\end{proof}

\section{Discussion}

\longversion{In this paper, w}\shortversion{W}e have introduced the notion of defensive alliances in signed networks and provided several algorithmic results for different problem settings. We conclude with suggesting further meaningful variants of problems concerning defensive alliances motivated from a practical perspective.
%We studied some basic algorithmic aspects of our notion, focussing on the simple variation that we defined.
%However, from a more practical perspective, several variants would be meaningful.
\begin{itemize}
\longversion{\item In our model, we considered negative relationships within an alliance in the same way as negative relationships outside, but with a different intuition\longversion{ and reasoning behind}: negative relations outside the alliance are meant to be ``enemies'', while negative relations inside should stay neutral, but we do not like to see too many of them also for psychological reasons; \longversion{or }if there are\longversion{ some} issues on which the alliance should agree, then it is also bad if one partner in the alliance has\longversion{ too} many negative relations inside. But one could think of treating both outside and inside negative relationships differently, possibly also in the sense of the next item.}
\longversion{\item In the literature of alliances, also a `parameterized variant' was discussed, requiring that the difference between (positive) relations inside and (negative) relations outside is at least~$r$ for some fixed~$r$. We only discussed the case $r=-1$ in this paper. The background is that, according to military history, an attacker should be really stronger than a defender to have a chance of victory. We did not consider this type of parameterizations at all so far.}
\item As in \textsc{Defensive Alliance Building}, counting the necessary flips means counting costs, it\longversion{would } make\shortversion{s} sense to introduce costs on the edges: some parties might be harder to persuade to become friends than others. 
Our alliance building algorithm can be adapted to this setting, as weighted variants of \textsc{UDCS} can be also solved in polynomial time.\longversion{\footnote{This is explained in more details in the technical report version~\cite{Gab82}.}}
    \item Likewise, in \textsc{Minimum Defensive Alliance}, it would make sense to put weights on the vertices, modelling different strengths of the involved parties. For instance, not all members of NATO \longversion{can be considered to be}\shortversion{are} equally powerful. %\todo{two-three sentences added for rev. 7 and two for rev. 6}
    In addition, edge weights in the interval $[-1,1]$ could be considered, better quantifying the degree of friendship or enmity\shortversion{ and, additionally, a cost function could indicate how costly it is to change the attitude of $u$ towards~$v$}. %\todo{HF: This would motivate seeing missing edges as neutral connections.} 
    Then, one could generalize the setting towards directed graphs, taking into consideration  that this degree need not be symmetric. \longversion{Similar notions have been introduced also for unsigned directed graphs, see \cite{MojSamYer2020}. Another possibility would be to associate (on top) costs or even cost functions to the variation with directed edges with edge weights, so that it becomes clear how much one has to pay to change the attitude of $u$ towards~$v$. }Most questions are open here.
    \item Finally, one may argue that it would be more justifiable not to treat ``internal'' and ``external'' enemies alike, as we do in our proposed definition of a defensive alliance~$S$. For instance, one could introduce a scaling factor~$\alpha\in [0,1]$ and require that $ \deg^+_S(v)+1\geq \alpha\cdot\deg^-_S(v)$ for all $v\in S$, so that $\alpha=1$ would correspond to the (first) condition discussed in this paper, while $\alpha=0$ means that ``internal enemies'' are completely ignored\shortversion{.}\longversion{, as they are supposed to stay neutral in case of a conflict, i.e., only the friends within an alliance count.} Clearly, setting $\alpha=0$ will take off the hardness of \shortversion{\textsc{Def\-All}.}\longversion{\textsc{Defensive Alliability}.} It %\todo{Discussion added to please rev. 6} 
    might be also interesting to consider factors $\alpha>1$, up to  directly requiring that all internal connections are positive. \longversion{However, notice that not all mutual connections within existing defensive alliances like NATO are genuinely friendly.\longversion{ For instance, the relation between Greece and Turkey has always been critical.} This is also true for the defensive alliances found in \cite{Rea54}.}
    Regarding \shortversion{\textsc{MinDefAll},}\longversion{\textsc{Minimum Defensive Alliance},} most hardness and algorithmic results hold for any $\alpha>0$.
\end{itemize}

\longversion{On a more technical side, a}\shortversion{A} very interesting \shortversion{technical} question is the complexity status of
\shortversion{\textsc{(Min)DefAll}}\longversion{\textsc{Defensive Alliability} or of \textsc{Minimum Defensive Alliance}} if the underlying graph is \emph{complete}. Recall that for the well-known question of \textsc{Correlation Clustering}, \NP-hardness is preserved under this restriction. %; it even has a name of its own: \textsc{Cluster Editing}. 
This question is related to asking if \textsc{(Min)DefAll} can be solved in \FPT-time when parameterized by the neighborhood diversity of the underlying unsigned graph, or even, more restrictively, by the \emph{distance to clique} (which denotes the vertex cover number of the complement graph).

To round up the picture of our paper, it would be also nice to know if \textsc{Defensive Alliability} is in \FPT when parameterized by treewidth or if \textsc{MinDefAll} is in \FPT when parameterized by the combined parameter treewidth plus solution size.
We pose both as open questions here.

Interestingly, after the first version of this report was published, a similar generalization of the notion of \emph{offensive alliances} towards signed graph was suggested in \cite{FenFMQ2023}; the according decision problems seem to be easier from the perspective of parameterized algorithms compared to defensive alliances.

%\section{Acknowledgments}
%\todo[color=magenta,inline]{No acknowledgments allowed in the submission}

%\newpage

\shortlong{\bibliographystyle{named}}{\bibliographystyle{apalike}}
\bibliography{ab,ijcai24}%,social_science}

\end{document}